\documentclass[10pt,twocolumn,journal]{IEEEtran}
\usepackage{latexsym}
\usepackage{amsfonts}
\usepackage{amsbsy}
\usepackage{amsmath,amssymb}
\usepackage{times}
\usepackage{graphicx}
\usepackage{enumerate}
\usepackage[usenames]{color}
\usepackage[dvips]{pstcol}
\usepackage{epstopdf}
\usepackage{subeqnarray}
\usepackage{cases}
\usepackage{stfloats}
\usepackage{cases}
\usepackage{subeqnarray}
\usepackage{amsmath}
\usepackage{multicol}
\usepackage{multirow}
\usepackage{bigstrut}
\usepackage{amsthm}
\usepackage{booktabs}
\usepackage{algorithm}
\usepackage{algorithmic}
\usepackage{bm}
\usepackage{amsthm}
\usepackage{cite}
\input epsf

\title{Two-timescale Beamforming Optimization for Intelligent Reflecting Surface Aided Multiuser Communication with QoS Constraints} 

\author{Ming-Min Zhao, \IEEEmembership{Member,~IEEE,} An Liu, \IEEEmembership{Senior Member,~IEEE,} Yubo Wan, and Rui Zhang, \IEEEmembership{Fellow,~IEEE}
	\thanks{
		M. M. Zhao, A. Liu and Y. Wan are with the College of Information Science and Electronic Engineering, Zhejiang University, China (email: \{zmmblack, anliu\}@zju.edu.cn, wan.scu@qq.com). R. Zhang is with the Department of Electrical and Computer Engineering, National University of Singapore, Singapore (email: elezhang@nus.edu.sg).
	}
}

\begin{document}
			\maketitle
	\begin{abstract}
Intelligent reflecting surface (IRS) is an emerging technology that is able to reconfigure the wireless channel via tunable passive signal reflection and thereby enhance the spectral/energy efficiency of wireless networks cost-effectively. In this paper, we study an IRS-aided multiuser multiple-input single-output (MISO) wireless system and adopt the two-timescale (TTS) transmission to reduce the signal processing complexity and channel training overhead as compared to the existing schemes based on the  instantaneous channel state information (I-CSI), and at the same time, exploit the multiuser channel diversity in transmission scheduling. Specifically, the \emph{long-term} passive beamforming (i.e., IRS phase shifts) is designed based on the statistical CSI (S-CSI) of all links, while the \emph{short-term} active beamforming (i.e., transmit precoding vectors at the access point (AP)) is designed to cater to the I-CSI of all users’ reconfigured channels with optimized IRS phase shifts. We aim to minimize the \emph{average transmit power} at the AP, subject to the users’ individual quality of service (QoS) constraints on the \emph{achievable long-term average rate}. The formulated stochastic optimization problem is non-convex and difficult to solve since the long-term and short-term design variables are complicatedly coupled in the QoS constraints. To tackle this problem, we propose an efficient algorithm, called the primal-dual decomposition based TTS joint active and passive beamforming (PDD-TJAPB), where the original problem is decomposed into a long-term passive beamforming problem and a family of short-term active beamforming problems, and the deep unfolding technique is employed to extract gradient information from the short-term problems to construct a convex surrogate problem for the long-term problem. We show that both the long-term and short-term problems can be efficiently solved and the proposed algorithm is proved to converge to a stationary solution of the original problem almost surely. Simulation results are presented which demonstrate the advantages and effectiveness of the proposed algorithm as compared to benchmark schemes.
		
	\end{abstract} 
	
	\section{Introduction} 
	Intelligent reflecting surface (IRS) has been regarded as a promising new technology to achieve high spectral efficiency with low hardware and energy cost for wireless communication  systems \cite{Wu2019Magazine, Basar2019, Huang2019, WuTutorial2020}. IRS is a man-made planar metasurface that is composed of a large number of passive reflecting elements, each of which is able to induce certain amplitude and/or phase changes to the impinging signal, thus collaboratively reconfiguring the wireless channel for enhanced communication. Compared with existing  wireless techniques such as massive multiple-input multiple-output (MIMO), ultra-dense network (UDN) and millimeter wave (mmWave) communication \cite{Boccardi2014}, IRS is more cost-effective since it can reduce the network deployment cost and energy consumption due to the ever-increasing number of active antennas/radio-frequency (RF) chains deployed in them. Thanks to its passive signal reflection, IRS can be densely deployed in wireless networks with a low and scalable cost as well as easily controllable interference. In addition, different from the traditional active relays, IRS can work in an inherent full-duplex manner without suffering any self-interference and additive noise.
	 
	 Due to the above appealing advantages, IRS has attracted significantly growing attention recently and it has been shown that IRS can enhance substantially the performance of various wireless systems (see, e.g., \cite{Yang2019, Cui2019, Jiang2019, wu2019joint, ZhangMIMO, lu2020enabling, zuo2020resource, Ding2020_IRSNOMA, PanJSAC2020, Pan2020TWC, Atapattu2020Twoway, PradhanWCL2020}). However, in order to reap the performance gains brought by IRS in practical systems, new challenges also arise. First, to fully exploit the new channel reconfiguration ability of IRS, the reflection coefficients at the IRS need to be adaptively designed based on the channel state information (CSI) of the IRS-reflected channels to/from both its associated access point (AP) and users. However, the acquisition of such CSI is a practically difficult task since IRS is a passive device without the ability to transmit pilot signals for channel estimation. In addition, the amount of channel coefficients in the IRS-reflected channels that need to be estimated increases rapidly with the number of IRS reflecting elements, especially when the AP and/or users are equipped with multiple antennas. To overcome this challenge, the authors in \cite{Hu2018TSP} proposed to deploy dedicated sensors at the IRS to estimate the IRS channels from the AP/users; however, this will increase the implementation cost of IRS and the estimated channels are also different from the desired IRS-reflected channels in general. Alternatively, estimating the AP-IRS-user cascaded channels based on the training/pilot signals sent by the AP/users can be a practical solution that also does not incur any performance loss as compared to the ideal case of knowing the individual AP-IRS and IRS-user channels \cite{Mishra2019ICASSP, zheng2019intelligent, you2019progressive}. Moreover, in \cite{wang2019channel, zheng2020intelligent}, the authors studied the channel estimation problem in IRS-aided multiuser systems and proposed new methods to reduce the channel training overhead by exploiting the fact that each IRS element reflects the signals from different users to the AP via the same IRS-AP channel. Besides, the low-rank and sparsity properties of IRS channels were exploited in \cite{He2019_CE, chen2019channel} for reducing the channel estimation complexity and overhead.
	 
Second, how to optimize the IRS reflection coefficients (including both reflection amplitudes and phase shifts in general) to unveil the full benefit offered by IRS is another challenging problem, due to their involved difficult-to-handle uni-modular constraints and intricate coupling in passive beamforming optimization. In the literature, various algorithms have been proposed to optimize the IRS reflection phase shifts (continuous or discrete), based on the instantaneous CSI (I-CSI) \cite{Yang2019, Cui2019, Jiang2019, wu2019joint, ZhangMIMO, lu2020enabling, zuo2020resource, Wu2018_journal, Wu2019Discrete, Guo2020TWC, Di2020JSAC}. In \cite{Zhao2020intelligent}, the authors further investigated the effect of IRS amplitude control on the passive beamforming performance, and showed that under imperfect CSI, additional performance gains can be achieved by jointly controlling the reflection amplitudes as compared to the case with full reflection and phase-shift control only.
However, the above works rely on the I-CSI, which requires high signal processing complexity and large channel training/signaling overhead for CSI acquisition. To tackle this difficulty, a two-timescale (TTS) transmission protocol was proposed in \cite{zhao2019intelligent} where the passive beamforming (IRS phase shifts) is optimized based on the statistical CSI (S-CSI) of all links, while  the active beamforming (transmit precoding vectors at the AP) is designed to cater to the user's reconfigured instantaneous channels with fixed (optimized) IRS passive beamforming. However, only the users' sum-rate maximization was considered in \cite{zhao2019intelligent} and the algorithms therein cannot be applied to the practical case when the users have individual rate or quality of service (QoS) requirements, which thus motivates this work.

In this paper, we study the TTS beamforming optimization problem in an IRS-aided multiuser system under the users’ individual rate/QoS constraints. The short-term active beamforming vectors at the AP and the long-term passive IRS reflection phase shifts are jointly designed to minimize the \emph{average transmit power} at the AP under the \emph{achievable long-term average rate constraints} for the users. Furthermore, we exploit the multiuser channel diversity to reduce the average transmit power of the AP, by properly designing the active beamforming vectors in different time slots according to the reconfigured I-CSI and performing ``time sharing'' based user scheduling to guarantee the long-term average rates of the users.

The formulated problem is a non-convex stochastic optimization problem and thus difficult to solve, with complicatedly coupled long-term and short-term design variables in the QoS constraints. To the best of our knowledge, there still lack efficient methods for solving such a TTS non-convex stochastic optimization problem. To overcome this challenge, we introduce long-term Lagrange multipliers to the QoS constraints and present a novel TTS primal-dual decomposition (PDD) method to decouple the long-term and short-term variables, which results in a long-term passive beamforming problem and a family of short-term active beamforming problems. Then, together with the constrained stochastic successive convex approximation (CSSCA) framework \cite{Liu_CSSCA_2019}, we propose a PDD-based TTS joint active and passive beamforming (PDD-TJAPB) algorithm to solve the formulated problem efficiently. The proposed algorithm contains two sub-algorithms, i.e., a long-term passive beamforming sub-algorithm and a short-term active beamforming sub-algorithm, to solve the long-term problem and each short-term problem, respectively. In the long-term sub-algorithm, a convex surrogate problem of the long-term problem is constructed based on randomly generated channel samples according to the S-CSI  and the gradient information of the short-term solutions with respect to (w.r.t.) the long-term variables. The resulting  problem is convex and thus can be efficiently solved using off-the-shelf solvers, such as CVX \cite{CVX}. In the short-term sub-algorithm, the weighted sum mean-square error minimization (WMMSE) method is employed to optimize the active beamforming vectors at the AP and a WMMSE network is built to extract the gradient information w.r.t. the long-term variables by unfolding the iterative procedure of WMMSE into a data flow graph and utilizing the widely-used back-propagation method. We prove that the solution obtained by the proposed PDD-TJAPB algorithm is a stationary point of the original problem almost surely. Finally, simulation results are presented to validate the effectiveness of the proposed algorithm.

The new contributions of this paper in view of the existing literature are summarized as follows:

1) A new TTS joint active and passive beamforming problem formulation for an IRS-aided multiuser MISO system is proposed to ensure the long-term QoS/rate constraint of each user and meanwhile exploit the multiuser channel diversity to reduce the average transmit power of the AP.

2) By combining a novel TTS primal-dual decomposition method, the CSSCA framework and the deep unfolding technique, we propose a new PDD-TJAPB algorithm to tackle the formulated problem, which cannot be solved by existing algorithms. We prove that every limit point of the proposed algorithm almost surely satisfies the stationary conditions of the formulated problem. The proposed algorithm provides an efficient solution to the open problem of TTS stochastic optimization with coupled long-term and short-term constraints.

3) Extensive simulation results are presented to validate the effectiveness of the proposed algorithm and we show that it can achieve considerable performance gain under various practical setups as compared to the conventional I-CSI-based scheme, with lower signal processing complexity and channel estimation/training overhead.

The rest of this paper is organized as follows. Section \ref{section_system_model} introduces the system model and problem formulation. In Section \ref{section_algorithm}, we propose an efficient PDD-TJAPB algorithm to solve the formulated problem and prove its convergence. Section \ref{section_simulation_results} presents numerical results to evaluate the performance of the
proposed algorithm and finally Section \ref{section_conclusion} concludes the paper.

\emph{Notations}: Scalars, vectors and matrices are respectively denoted by lower/upper case, boldface lower case and boldface upper case letters. For an arbitrary matrix $\mathbf{A}$, $\mathbf{A}^T$ and $\mathbf{A}^{H}$ denote its transpose and conjugate transpose, respectively. For a square matrix $\mathbf{A}$, $\mathbf{A}^{-1}$ denotes its inverse matrix if it is invertible. $[\mathbf{a}]_n$ denotes the $n$-th element of a vector $\mathbf{a}$, and $\mathbf{a} \succeq \mathbf{b}$ denotes the component-wise inequality between vectors $\mathbf{a}$ and $\mathbf{b}$. $\|\cdot\|$ denotes the Euclidean norm of a complex vector and $|\cdot|$ denotes the absolute value of a complex scalar. $\mathcal{CN}(\mathbf{x},\bm{\Sigma})$ denotes the distribution of a circularly symmetric complex Gaussian (CSCG) random vector with mean vector $\mathbf{x}$ and covariance matrix $\bm{\Sigma}$; and $\sim$ stands for ``distributed as''. For any given $x_1,\cdots,x_N$, $\textrm{diag}(x_1,\cdots,x_N)$ denotes a diagonal matrix with $x_1,\cdots,x_N$ as its diagonal elements. The letter $\jmath$ will be used to represent $\sqrt{-1}$. $\mathbb{C}^{n\times m}$ ($\mathbb{R}^{n\times m}$) denotes the space of $n\times m$ complex (real) matrices and $\mathbb{R}^{n++}$ represents the space of $n\times 1$ vectors with strictly positive real elements. $\mathcal{F}^N$ is defined as the Cartesian product of $N$ identical sets each given by $\mathcal{F}$. $\textrm{relint}(\mathcal{S})$ denotes the relative interior of a set $\mathcal{S}$, defined as its interior within the affine hull of $\mathcal{S}$. For any nonempty convex set $\mathcal{S}$, $\textrm{relint}(\mathcal{S})\triangleq \{x\in\mathcal{S}: \forall y \in \mathcal{S}, \exists \lambda>1 : \lambda x + (1-\lambda)y \in \mathcal{S} \}$. $\mathbf{I}$, $\mathbf{1}$ and $\mathbf{0}$ denote an identity matrix, an all-one matrix/vector and an all-zero matrix with appropriate dimensions, respectively. $\mathbb{E}\{\cdot\}$ represents the statistical expectation. For given two sets $\mathcal{A}$ and $\mathcal{B}$, $\mathcal{A} \backslash \mathcal{B}\triangleq \{x|x \in \mathcal{A}, x \notin \mathcal{B} \}$.

\section{System Model and Problem Formulation} \label{section_system_model}
\subsection{System Model}
As shown in Fig. \ref{system_model}, we consider an IRS-aided multiuser MISO downlink system, where an IRS composed of $N$ reflecting elements is deployed to assist the communications from an AP with $M$ antennas to a set of $K$ single-antenna users (denoted by $\mathcal{K}\triangleq\{1,\cdots, K\}$). In the considered system, the IRS is connected to a smart controller that is able to communicate with the AP via a separate wireless link for coordinating their transmission and exchanging information, such as CSI and IRS reflection coefficients \cite{Wu2019Magazine, Zhang2018Metasurface}. Since the signal transmitted through the AP-IRS-user link experiences severe ``distance-product'' power loss \cite{Wu2019Magazine}, we only consider the signal reflected by the IRS for the first time and those reflected by it two or more times are ignored.

Let $\mathbf{G} \in \mathbb{C}^{N\times M}$ denote the baseband equivalent channel matrix of the AP-IRS link, $\mathbf{h}_{r,k} \in \mathbb{C}^{N\times 1}$ and $\mathbf{h}_{d,k} \in \mathbb{C}^{M\times 1}$ denote the (conjugate) channel vectors of the IRS-user $k$ and AP-user $k$ links, respectively. Then, the received signal of user $k$ can be expressed as
\begin{equation}
y_k = (\mathbf{h}_{r,k}^H \mathbf{\bm{\Theta}} \mathbf{G}+\mathbf{h}_{d,k}^H) \mathbf{x} + n_k,
\end{equation}
where $\bm{\Theta}$ denotes an $N \times N$ diagonal reflection coefficient matrix (also known as the passive beamforming matrix), which can be written as $\bm{\Theta} = \textrm{diag}(\phi_1,\cdots,\phi_N )$, $\phi_n = a_n e^{\jmath \theta_n}$ with $a_n \in [0,1] $ and $\theta_n \in \mathcal{F} \triangleq [0,2\pi), \forall n \in \mathcal{N} \triangleq\{1,\cdots,N\}$, denoting the reflection amplitudes and phase shifts, respectively. 
For simplicity, we assume $a_n = 1$, $\forall n \in \mathcal{N}$ to maximize the signal reflection \cite{zhao2019intelligent}. $n_k$ denotes the independent and identically distributed (i.i.d.) additive white Gaussian noise (AWGN) at the receiver of user $k$ with zero-mean and variance $\sigma_k^2$. $\mathbf{x}$ denotes the complex baseband transmitted signal at the AP, which can be written as
\begin{equation}
\mathbf{x} = \sum\limits_{k\in\mathcal{K}}\mathbf{w}_k s_k,
\end{equation}
where $s_k \sim \mathcal{CN}(0,1)$ represents the information-bearing symbol for user $k$ with $s_k$'s assumed to be i.i.d., and $\mathbf{w}_k \in \mathbb{C}^{M \times 1}$ denotes the active beamforming/precoding vector intended for user $k$. 

\begin{figure}[t] 
	\setlength{\abovecaptionskip}{-0.2cm} 
	\setlength{\belowcaptionskip}{-0.2cm} 
	\centering
	\scalebox{0.4}{\includegraphics{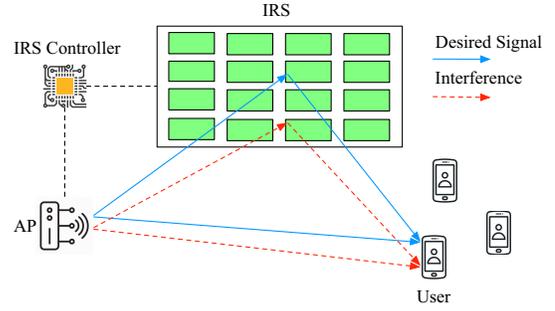}} 
	\caption{An IRS-aided multiuser MISO downlink system.}
	\label{system_model}
\end{figure}

Then, the received signal-to-interference-plus-noise ratio (SINR) and achievable rate of user $k$ in bits/second/Hertz (bits/s/Hz) can be respectively written as
\begin{equation} \label{SINR}
\textrm{SINR}_k(\bm{\Theta},\{\mathbf{w}_k\},\tilde{\mathbf{H}})\!=\! \frac{|(\mathbf{h}_{r,k}^H \bm{\Theta} \mathbf{G}\!+\! \mathbf{h}_{d,k}^H) \mathbf{w}_k|^2}{\sum\limits_{j \in \mathcal{K}\backslash\{k\}}|(\mathbf{h}_{r,k}^H \bm{\Theta} \mathbf{G}\!+\! \mathbf{h}_{d,k}^H) \mathbf{w}_j|^2 \!+\! \sigma_k^2},
\end{equation}
\begin{equation} \label{I_rate}
r_k(\bm{\Theta},\{\mathbf{w}_k\},\tilde{\mathbf{H}}) = \log(1+\textrm{SINR}_k(\bm{\Theta},\{\mathbf{w}_k\},\tilde{\mathbf{H}})),
\end{equation}
 where $ \tilde{\mathbf{H}} \triangleq \{\mathbf{G}, \mathbf{h}_{r,k},\mathbf{h}_{d,k}\}$ denotes the full channel ensemble. Since optimizing $\bm{\Theta}$ is equivalent to optimizing the phase of each of its diagonal elements, i.e., the phase shift of each reflecting element, by letting $\bm{\theta} = [\theta_1,\cdots, \theta_N]^T$, the SINR and achievable rate of user $k$ can be equivalently expressed as $\textrm{SINR}_k(\bm{\theta},\{\mathbf{w}_k\},\tilde{\mathbf{H}})$ and $r_k(\bm{\theta},\{\mathbf{w}_k\},\tilde{\mathbf{H}})$, without loss of optimality.\footnote{Note that with this variable substitution, the non-convex uni-modular constraints on the reflection coefficients $|\phi_n|=1,\; \forall n \in \mathcal{N}$ are transformed to $\theta_n \in [0,2\pi) ,\; \forall n \in \mathcal{N}$, which are convex and easier to handle.}

\subsection{Problem Formulation}
In this paper, we aim to minimize the average transmit power at the AP by jointly optimizing the active beamforming vectors at the AP and passive beamforming (reflection phase shifts) at the IRS, subject to the constraints on the achievable average rates for all users.
In order to reduce the signal processing complexity and training/signaling overhead of the joint active and passive beamforming design based on the I-CSI of all links, we adopt the TTS transmission protocol proposed in \cite{zhao2019intelligent}. Specifically, we focus on a sufficiently long time interval within which the statistics (distributions) of all channels are assumed to be constant as we mainly consider low-mobility users.\footnote{How to efficiently estimate the CSI (either instantaneously or statistically) and design the active and passive beamforming vectors in the high-mobility scenario are still open problems in this field.} The considered time interval consists of $T_s \gg 1$ time slots and all channels are assumed to remain approximately constant within each time slot, while that may change over different time slots based on the assumed constant channel distributions. We further assume that the AP can obtain the effective I-CSI $\{ \tilde{\mathbf{h}}_k \triangleq  \mathbf{G}^H \bm{\Theta}^H  \mathbf{h}_{r,k}+ \mathbf{h}_{d,k} \} $ at each time slot with given $\bm{\Theta}$ by conventional channel estimation techniques (e.g., estimating the reciprocal uplink channels based on the uses' training signals), and the S-CSI of the AP-IRS, AP-user and IRS-user links is known at the AP at the beginning of each time interval \cite{zhao2019intelligent}. Therefore, we aim to optimize the passive IRS phase-shift matrix $\bm{\Theta}$ based on the S-CSI of all links, while the active beamforming vectors $\{\mathbf{w}_k \}$ at the AP are designed to cater to the I-CSI of the users’ effective fading channels $\{ \tilde{\mathbf{h}}_k\}$ with the optimized IRS phase shifts that are fixed over all the time slots. 

\newtheorem{remark}{Remark}
\begin{remark}
\emph{For completeness, we briefly introduce the S-CSI estimation approach in \cite{zhao2019intelligent}, where the S-CSI between the IRS and the AP/users are estimated by resorting to some dedicated sensors/receiving circuits at the IRS and leveraging the pilots and/or data transmitted in both uplink and downlink using standard mean and covariance matrices estimation techniques. Specifically, we can obtain $T_p$ channel samples, i.e., $\{\mathbf{h}_{d,k}^i\}$, $\{\mathbf{G}^i\}$ and $\{\mathbf{h}_{r,k}^i\}$, $i=1,\cdots,T_p$, through conventional pilot-based channel estimation methods.  Then, we are able to estimate the S-CSI of the AP-user, AP-IRS, and IRS-user links based on these channel samples by utilizing the methods in \cite{Mestre2008}, \cite{Werner2008}. Typically, the number of required channel samples $T_p$ (directly proportional to the required pilot overhead) is comparable to the dimensions of the channels, however for some special cases, e.g., for sparse or low-rank channels, $T_p$ can be much smaller \cite{Liu2015TWC}. On the other hand, for fully-passive IRSs without any sensors/receiving circuits, how to effectively estimate the S-CSI of the cascaded AP-IRS-user channels is an important problem that is not addressed yet and remains open in the IRS literature.}
\end{remark}

Note that although the IRS phase shifts can be adjusted in almost real time (e.g., for cascaded channel estimation), they are fixed in the data transmission phase of the adopted TTS transmission protocol, and vary slightly in consecutive time intervals since the long-term S-CSI varies slowly. Besides, the number of reflecting elements at the IRS, $N$, is typically much larger than that of transmit antennas at the AP, thus the effective CSI $\{ \tilde{\mathbf{h}}_k\}$ usually has a much smaller dimension than $ \tilde{\mathbf{H}} $ and estimating $\{ \tilde{\mathbf{h}}_k\}$ requires less training/pilot symbols as compared to estimating $ \tilde{\mathbf{H}} $ (or the AP-IRS-user cascaded channels $\{\textrm{diag}(\mathbf{h}_{r,k}^H) \mathbf{G}\}$ and AP-user direct channels $\{\mathbf{h}_{d,k}\}$ equivalently). Specifically, the number of channel coefficients required in each time slot can be reduced from $NM+NK+MK$ to $MK$ by employing the proposed TTS transmission protocol and the number of IRS reflection coefficients that are needed to be transmitted to the IRS through the separate wireless link in one time interval can also be reduced from $T_sN$ to $N$.

Consequently, we formulate the following QoS-constrained TTS optimization problem:
\begin{subequations} \label{TTS_problem_detail_equi}
	\begin{align}
	\min\limits_{\bm{\theta},\; \{\mathbf{w}_k(\tilde{\mathbf{h}}),\forall \tilde{\mathbf{h}}\}} \; &  \mathbb{E} \left\{\sum\limits_{k\in \mathcal{K}} \|\mathbf{w}_k (\tilde{\mathbf{h}})\|^2 \right\}\\
	\textrm{s.t.} \; &
	\mathbb{E} \{  r_k(\bm{\theta},\{\mathbf{w}_k(\tilde{\mathbf{h}})\},\tilde{\mathbf{H}} ) \} \geq R_k, \;\forall k \in \mathcal{K}, \label{long_term_constraints}\\
	& \theta_n \in \mathcal{F} ,\; \forall n \in \mathcal{N}.
	\end{align}
\end{subequations}
where the full channel ensemble $\tilde{\mathbf{H}} $ can be viewed as a random system state defined on the probability space $(\Omega, \mathcal{G},\mathbb{P})$, with $\Omega$
being the sample space, $\mathcal{G}$ being the $\sigma$-algebra generated by subsets of $\Omega$, and $\mathbb{P}$ being a probability measure defined on $\mathcal{G}$; {the notation $\{\mathbf{w}_k(\tilde{\mathbf{h}})\}$, where $\tilde{\mathbf{h}} \triangleq \{\tilde{\mathbf{h}}_k \}$, denotes the set of short-term active beamforming vectors under the effective CSI $\tilde{\mathbf{h}}$ (or equivalently, the system state $\tilde{\mathbf{H}}$) at each time slot}; $R_k > 0$ is the minimum average rate target of user $k$ and the expectations in problem \eqref{TTS_problem_detail_equi} are taken over the distribution of $\{\tilde{\mathbf{H}}\}$.

Problem \eqref{TTS_problem_detail_equi} is challenging to solve because: 1) the short-term beamforming vectors $\{\mathbf{w}_k \}$ and the long-term IRS phase shifts $\bm{\theta}$ are intricately coupled in a complicated manner via the long-term QoS constraints in \eqref{long_term_constraints}; 2) a closed-form expression of the achievable average rate of each user, $\mathbb{E} \{  r_k(\bm{\theta},\{\mathbf{w}_k(\tilde{\mathbf{h}})\},\tilde{\mathbf{H}}) \}$, is difficult to obtain. Generally, there is no efficient method for solving the non-convex stochastic optimization problem \eqref{TTS_problem_detail_equi} optimally. In the next section, we propose an efficient algorithm, called PDD-TJAPB, to solve it sub-optimally and with the guaranteed convergence to a stationary solution to problem \eqref{TTS_problem_detail_equi}, which is defined as follows.
\newtheorem{definition}{Definition} 
\begin{definition}[Stationary solution of problem \eqref{TTS_problem_detail_equi}]
\emph{A solution $(\bm{\theta}^* \in \mathcal{F}^N, \bm{\varUpsilon}^* = \{ \mathbf{w}^*(\tilde{\mathbf{h}}) \in \mathbb{C}^{MK\times 1}, \forall \tilde{\mathbf{H}}  \})$ where $\mathbf{w}^*(\tilde{\mathbf{h}}) \triangleq [\mathbf{w}_1^*(\tilde{\mathbf{h}})^T,\cdots, \mathbf{w}_K^*(\tilde{\mathbf{h}})^T]^T$ is called a stationary solution of problem \eqref{TTS_problem_detail_equi}, if there exist long-term Lagrange multipliers $\bm{\lambda}=[\lambda_1,\cdots, \lambda_K]^T \succeq \mathbf{0}$ associated with the achievable  average rate constraints \eqref{long_term_constraints}, such that the following conditions are satisfied:}

\emph{
1) For every $\tilde{\mathbf{H}} \in \Omega$, we
have
\begin{equation} \label{stationary_1}
\begin{aligned}
\partial_{\mathbf{w}} \Bigg(\sum\limits_{k\in \mathcal{K}} & \|\mathbf{w}^*(\tilde{\mathbf{h}})\|^2 \Bigg) \\
 & + \sum\limits_{k \in \mathcal{K}}
\lambda_k \partial_{\mathbf{w}}(R_k - r_k(\bm{\theta}^*,\mathbf{w}^*(\tilde{\mathbf{h}}),\tilde{\mathbf{H}} ) ) = \mathbf{0},
\end{aligned}
\end{equation}
where $\partial_{\mathbf{w}} (\sum\nolimits_{k\in \mathcal{K}}  \|\mathbf{w}^*(\tilde{\mathbf{h}})\|^2)$ and $\partial_{\mathbf{w}}(R_k - r_k(\bm{\theta}^*,\mathbf{w}^*(\tilde{\mathbf{h}}),\tilde{\mathbf{H}} ) )$ are the partial derivatives of the functions $\sum\nolimits_{k\in \mathcal{K}}  \|\mathbf{w}(\tilde{\mathbf{h}})\|^2$ and $ R_k - r_k(\bm{\theta}^*,\mathbf{w}(\tilde{\mathbf{h}}),\tilde{\mathbf{H}} )  $ w.r.t. $\mathbf{w}$ at $\mathbf{w} = \mathbf{w}^*(\tilde{\mathbf{h}})$.}

\emph{
2)
\begin{equation} \label{stationary_2}
\begin{aligned}
 \sum\limits_{k \in \mathcal{K}} \lambda_k\partial_{\bm{\theta}}(R_k - \mathbb{E}\{ r_k(\bm{\theta}^*,\mathbf{w}^*(\tilde{\mathbf{h}}),\tilde{\mathbf{H}} ) \}) & = \mathbf{0},\\
 R_k - \mathbb{E}\{ r_k(\bm{\theta}^*,\mathbf{w}^*(\tilde{\mathbf{h}}),\tilde{\mathbf{H}} ) \}  & \leq 0, \; \forall k \in \mathcal{K},
\end{aligned}
\end{equation}
where $\partial_{\bm{\theta}}(R_k - \mathbb{E}\{ r_k(\bm{\theta}^*,\mathbf{w}^*(\tilde{\mathbf{h}}),\tilde{\mathbf{H}} ) \})$ is the partial derivative of $ R_k - \mathbb{E}\{ r_k(\bm{\theta},\mathbf{w}^*(\tilde{\mathbf{h}}),\tilde{\mathbf{H}} ) \} $ w.r.t. $\bm{\theta}$ at $\bm{\theta} = \bm{\theta}^*$.}

\emph{
3) 
\begin{equation} \label{stationary_3}
\lambda_k (R_k - \mathbb{E}\{ r_k(\bm{\theta}^*,\mathbf{w}^*(\tilde{\mathbf{h}}),\tilde{\mathbf{H}} ) \}) = 0, \; \forall k \in \mathcal{K}.
\end{equation}
}
\end{definition}

In other words, a stationary solution satisfies all the Karush–Kuhn–Tucker (KKT) conditions of problem \eqref{TTS_problem_detail_equi}. Note that a stationary solution is not necessarily a global optimal or local optimal solution, however, obtaining a stationary solution is perhaps the best we can do for problem \eqref{TTS_problem_detail_equi} with an acceptable computational complexity, since it is a non-convex stochastic optimization problem with coupled constraints.

Note that in \cite{Wu2018_journal}, a similar QoS-constrained power minimization problem for an IRS-aided multiuser system was investigated, where $\{\mathbf{w}_k \}$ and $\bm{\theta}$ are designed based on the I-CSI. This work is different from \cite{Wu2018_journal} since we adopt a TTS beamforming design and the minimum \emph{achievable average rate} of each user is considered as the constraint, instead of the \emph{instantaneous achievable rate} in \cite{Wu2018_journal}. The long-term IRS phase shifts in $\bm{\theta}$ are designed based on the S-CSI while the short-term active beamforming vectors are also designed in a different way. Moreover, due to the average rate constraints for the users, their instantaneous achievable rates are not required to be larger than or equal to $R_k$ in each time slot; therefore, the multiuser channel diversity can be exploited to achieve efficient variable-rate scheduling over different time slots for the users in our design, which is intuitively illustrated in Fig. \ref{multiuser_channel_diversity}. As can be seen, in the TTS transmission protocol, the instantaneous achievable rates of the users can be varying according to their channel conditions in each time slot, thus the transmit power of the AP can be more efficiently allocated over time to save energy. Simulation results will be provided to further illustrate this gain in Section \ref{section_simulation_results}. On the other hand, although we consider the same TTS transmission protocol in this paper as that in \cite{zhao2019intelligent}, the investigated optimization problems are different since we consider long-term QoS constraints in this paper, while in \cite{zhao2019intelligent} the users' sum-rate maximization was considered. Note that in general, the formulation in \cite{zhao2019intelligent} cannot guarantee the average achievable rate for each user. {Such QoS requirements are quite common in practical communication scenarios as different users can have very diverse demands for wireless data services, such as augmented reality (AR)/virtual reality (VR), video streaming and text messaging, etc.} Moreover, problem \eqref{TTS_problem_detail_equi} is more difficult to solve due to the newly considered users' long-term rate constraints and the algorithms in \cite{zhao2019intelligent} are no longer applicable.

\begin{remark}
\emph{It would be more practical to further include the average and peak power constraints in problem \eqref{TTS_problem_detail_equi} \cite{Khojastepour2004}, i.e., $\mathbb{E} \{\sum\nolimits_{k\in \mathcal{K}} \|\mathbf{w}_k (\tilde{\mathbf{h}})\|^2 \} \leq P_{\textrm{ave}}$ and $\sum\nolimits_{k\in \mathcal{K}} \|\mathbf{w}_k (\tilde{\mathbf{h}})\|^2 \leq P_{\textrm{peak}}, \forall \tilde{\mathbf{h}}$, where $P_{\textrm{ave}}$ and $P_{\textrm{peak}}$ denote the average and peak power limits at the AP, respectively. The proposed algorithm can be modified to solve problem \eqref{TTS_problem_detail_equi} in the presence of the average and/or peak power constraints since the former can be similarly handled as constraint \eqref{long_term_constraints}, while the latter is convex and can be tackled by the WMMSE method \cite{Shi2011WMMSE}. }
\end{remark}

\begin{figure*}[t] 
	\setlength{\abovecaptionskip}{-0.2cm} 
	\setlength{\belowcaptionskip}{-0.2cm} 
	\centering
	\scalebox{0.47}{\includegraphics{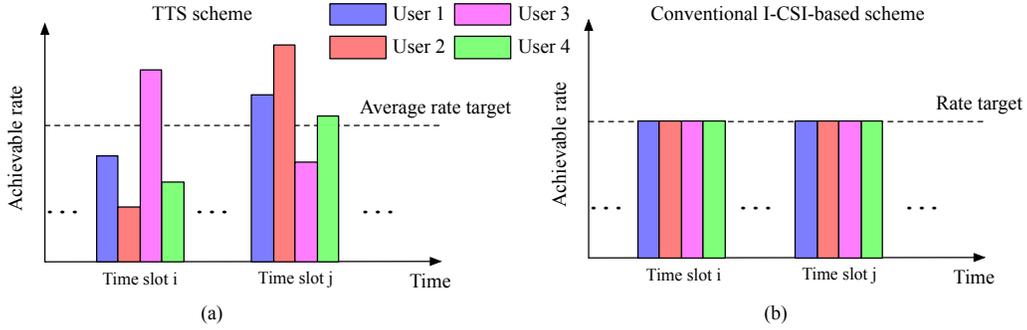}} 
	\caption{Illustration of the users' achievable rates in the TTS and conventional I-CSI-based schemes.}
	\label{multiuser_channel_diversity}
\end{figure*}

\begin{remark}
\emph{In problem \eqref{TTS_problem_detail_equi}, we have assumed for simplicity that the S-CSI and effective I-CSI are perfectly known at the AP at the beginning of the considered time interval and each of its time slots, respectively. In practice, however, channel estimation errors are inevitable due to the limited channel training resources (such as power and time) \cite{Zhao2020intelligent, zhou2020framework}. Further investigation into the TTS design under imperfect S-CSI and effective I-CSI is left for future work.}
\end{remark}

\section{Proposed PDD-TJAPB Algorithm} \label{section_algorithm}
In this section, we propose the PDD-TJAPB algorithm to solve problem \eqref{TTS_problem_detail_equi}, where a novel TTS primal-dual decomposition method is proposed to decouple the variables for efficient optimization. Specifically, we decompose problem \eqref{TTS_problem_detail_equi} into a long-term passive beamforming problem and a family of short-term active beamforming problems (each corresponds to one channel sample in $\{\tilde{\mathbf{H}}\}$). Each short-term problem is solved efficiently by employing the WMMSE method \cite{Shi2011WMMSE} to obtain a stationary point, while the long-term problem is tackled based on the CSSCA framework \cite{Liu_CSSCA_2019}, where randomly generated channel samples in $\{\tilde{\mathbf{H}}\}$ are utilized to construct concave surrogate functions for the achievable average rates $\mathbb{E} \{  r_k(\bm{\theta},\{\mathbf{w}_k(\tilde{\mathbf{h}})\}, \tilde{\mathbf{H}}) \}$, $\forall k \in \mathcal{K}$ and convex ones for the objective function $\mathbb{E} \{\sum\nolimits_{k\in \mathcal{K}} \|\mathbf{w}_k (\tilde{\mathbf{h}})\|^2 \}$. In particular, when constructing these surrogate functions, the deep unfolding technique \cite{hershey2014deep} is used to extract the gradient information of the short-term solutions $\{\mathbf{w}_k(\tilde{\mathbf{h}})\}$ obtained by the WMMSE algorithm w.r.t. the long-term (primal and dual) variables $\bm{\theta}$ and $\bm{\lambda}$.

\subsection{TTS Primal-Dual Decomposition for Problem \eqref{TTS_problem_detail_equi}}
In this subsection, we prove a novel TTS primal-dual decomposition method to decouple the optimization variables. Specifically, for fixed long-term phase shifts $\bm{\theta}$ and long-term Lagrange multipliers $\bm{\lambda}$, let $\mathbf{w}^J(\tilde{\mathbf{h}}(\bm{\theta},\tilde{\mathbf{H}}), \bm{\lambda}) = [\mathbf{w}_1^J(\tilde{\mathbf{h}}(\bm{\theta},\tilde{\mathbf{H}}), \bm{\lambda})^T ,\cdots, \mathbf{w}_K^J(\tilde{\mathbf{h}}(\bm{\theta},\tilde{\mathbf{H}}), \bm{\lambda})^T ]^T $ denote a stationary solution of the following \emph{short-term subproblem}:\footnote{Note that mathematically $\mathcal{P}_S(\tilde{\mathbf{h}}(\bm{\theta},\tilde{\mathbf{H}}), \bm{\lambda} )$ and $\mathbf{w}^J(\tilde{\mathbf{h}}(\bm{\theta},\tilde{\mathbf{H}}), \bm{\lambda})$ depend on $\tilde{\mathbf{H}}$, however since $\bm{\theta}$ is fixed when solving $\mathcal{P}_S(\tilde{\mathbf{h}}(\bm{\theta},\tilde{\mathbf{H}}), \bm{\lambda} )$, only the knowledge of the effective CSI $\{\tilde{\mathbf{h}}_k\}$ is required in practical implementation, which will be elaborated in Section \ref{sec_short_term}.}
\begin{equation} \label{short_term_problem}
\begin{aligned}
\mathcal{P}_S&(\tilde{\mathbf{h}}(\bm{\theta},\tilde{\mathbf{H}}), \bm{\lambda} ):\\
& \min \limits_{\{\mathbf{w}_k\}} \;  \sum\limits_{k\in \mathcal{K}} \|\mathbf{w}_k\|^2 + \sum\limits_{k \in \mathcal{K}}\lambda_k (R_k-r_k(\bm{\theta},\{\mathbf{w}_k\}, \tilde{\mathbf{H}})),
\end{aligned}
\end{equation}
obtained by running the short-term active beamforming sub-algorithm (based on the WMMSE method) for a sufficiently large ($J$) iterations. $\mathbf{w}^J(\tilde{\mathbf{h}}(\bm{\theta},\tilde{\mathbf{H}}), \bm{\lambda})$ thus satisfies the following KKT condition of $\mathcal{P}_S(\tilde{\mathbf{h}}(\bm{\theta},\tilde{\mathbf{H}}), \bm{\lambda} )$: 
\begin{equation} \label{KKT_short_term}
\begin{aligned}
\Bigg\| \partial_{\mathbf{w}} \Bigg(\sum\limits_{k\in \mathcal{K}} &  \|\mathbf{w}^J(\tilde{\mathbf{h}}(\bm{\theta},\tilde{\mathbf{H}}), \bm{\lambda})\|^2\Bigg) + \sum\limits_{k \in \mathcal{K}}
\lambda_k \partial_{\mathbf{w}}(R_k  \\
& - r_k(\bm{\theta},\mathbf{w}^J(\tilde{\mathbf{h}}(\bm{\theta},\tilde{\mathbf{H}}), \bm{\lambda}),\tilde{\mathbf{H}} ) ) \Bigg\| = e^J(\bm{\theta}, \bm{\lambda}),
\end{aligned}
\end{equation}
where $e^J(\bm{\theta}, \bm{\lambda})$ denotes the error due to the fact that the short-term sub-algorithm only runs for a finite number of $J$ iterations. It will be shown in Section \ref{sec_short_term} that as $J \rightarrow \infty$, the solution $\{ \mathbf{w}_k^J(\tilde{\mathbf{h}}(\bm{\theta},\tilde{\mathbf{H}}), \bm{\lambda}) \}$ obtained by running the proposed short-term sub-algorithm converges to a stationary point of $\mathcal{P}_S(\tilde{\mathbf{h}}(\bm{\theta},\tilde{\mathbf{H}}), \bm{\lambda} )$, $\forall \bm{\theta} \in \mathcal{F}^N$, $\bm{\lambda} \succeq \mathbf{0}$, i.e., $\lim_{J \rightarrow \infty}e^J(\bm{\theta}, \bm{\lambda}) = 0$. However, in practice, the short-term sub-algorithm always runs for a finite number of iterations, thus it is necessary to derive a new TTS primal-dual decomposition theorem for this case.

With the short-term policy $\{\mathbf{w}_k^J(\tilde{\mathbf{h}}(\bm{\theta},\tilde{\mathbf{H}}), \bm{\lambda})\}$, we formulate the following long-term problem:
\begin{equation} \label{long-term_problem_relax}
\begin{aligned}
\mathcal{P}_L^J:\quad\quad\quad& \\
 \min\limits_{\bm{\theta} \in \mathcal{F}^N,\;\bm{\lambda} \succeq \mathbf{0}} \; &  \mathbb{E} \left\{\sum\limits_{k\in \mathcal{K}} \|\mathbf{w}_k^J(\tilde{\mathbf{h}}(\bm{\theta},\tilde{\mathbf{H}}), \bm{\lambda})  \|^2 \right\}\\
\textrm{s.t.} \; & \mathbb{E} \{  r_k(\bm{\theta},\{ \mathbf{w}_k^J(\tilde{\mathbf{h}}(\bm{\theta},\tilde{\mathbf{H}}), \bm{\lambda})  \},\tilde{\mathbf{H}} ) \} \! \geq \! R_k, \;\forall k \in \mathcal{K}.
\end{aligned}
\end{equation}
Let $\nabla_{\bm{\theta}}  \mathbb{E} \{\sum\nolimits_{k\in \mathcal{K}} \|\mathbf{w}_k^J(\tilde{\mathbf{h}}(\bm{\theta},\tilde{\mathbf{H}}), \bm{\lambda})  \|^2 \} \triangleq \mathbb{E} \{  \partial_{\mathbf{\bm{\theta}}}  \mathbf{w}^{J}( \tilde{\mathbf{h}} (  \bm{\theta}, \tilde{\mathbf{H}} ), \allowbreak  \bm{\lambda})  \partial_{\mathbf{w}}   (\sum\nolimits_{k\in \mathcal{K}}   \|\mathbf{w}_k^J(   \tilde{\mathbf{h}}( \bm{\theta},  \tilde{\mathbf{H}}),   \bm{\lambda})\|^2 )  \} $ and 
$\nabla_{\bm{\lambda}}    \mathbb{E} \{  \sum\nolimits_{k\in \mathcal{K}}
 \|  \mathbf{w}_k^J(  \allowbreak  \tilde{\mathbf{h}}(  \bm{\theta},  \tilde{\mathbf{H}} ),   \bm{\lambda})  \|^2   \}   \triangleq \mathbb{E} \{   \partial_{\mathbf{\bm{\lambda}}}   \mathbf{w}^{J}(  \tilde{\mathbf{h}}( \bm{\theta},  \tilde{\mathbf{H}}),   \bm{\lambda}  ) \partial_{\mathbf{w}} (  \sum\nolimits_{k\in \mathcal{K}} \|   \mathbf{w}_k^J(  \tilde{\mathbf{h}}( \allowbreak \bm{\theta},\tilde{\mathbf{H}}),   \bm{\lambda})\|^2 )  \} $  denote the derivatives of $\sum\nolimits_{k\in \mathcal{K}} \|\mathbf{w}_k^J(\tilde{\mathbf{h}}(\bm{\theta},\allowbreak  \tilde{\mathbf{H}}), \bm{\lambda})\|^2$ w.r.t. $\bm{\theta}$ and $\bm{\lambda}$, respectively, where $ \partial_{\bm{\theta}} \mathbf{w}^J(\tilde{\mathbf{h}}( \allowbreak \bm{\theta}, \allowbreak\tilde{\mathbf{H}}), \allowbreak \bm{\lambda}) \in \mathbb{C}^{N  \times (MK)}$ ($\partial_{\bm{\lambda}} \mathbf{w}^J(\tilde{\mathbf{h}}(\bm{\theta},\tilde{\mathbf{H}}), \bm{\lambda}) \in \mathbb{C}^{K \times (MK)}$) is the derivative of the vector function $ \mathbf{w}^J(\tilde{\mathbf{h}}(\bm{\theta},\tilde{\mathbf{H}}),  \bm{\lambda})$ w.r.t. the vector $\bm{\theta}$ ($\bm{\lambda}$).  
Similarly, let $\nabla_{\bm{\theta}}(R_k - \mathbb{E} \{ r_k(\bm{\theta}, \mathbf{w}^J(\tilde{\mathbf{h}}( \bm{\theta}, \tilde{\mathbf{H}}), \allowbreak \bm{\lambda}) ,\tilde{\mathbf{H}} ) \} ) \triangleq  \mathbb{E}\{ \partial_{\mathbf{\bm{\theta}}} \mathbf{w}^{J}(\tilde{\mathbf{h}}(\bm{\theta},\allowbreak \tilde{\mathbf{H}}), \bm{\lambda})  \partial_{\mathbf{w}} (R_k -  r_k(\bm{\theta}, \mathbf{w}^J(\tilde{\mathbf{h}}(\bm{\theta}, \allowbreak \tilde{\mathbf{H}}), \bm{\lambda}) ,\tilde{\mathbf{H}} ) )  + \partial_{\bm{\theta}}(R_k -  \allowbreak r_k(\bm{\theta}, \allowbreak \mathbf{w}^J(\tilde{\mathbf{h}}(\bm{\theta}, \tilde{\mathbf{H}}), \bm{\lambda}),\tilde{\mathbf{H}} ) )  \}$  and $\nabla_{\bm{\lambda}} ( \allowbreak R_k - \mathbb{E} \{ r_k(\bm{\theta}, \mathbf{w}^J(\tilde{\mathbf{h}}(\bm{\theta},\allowbreak \tilde{\mathbf{H}}), \bm{\lambda}) ,\tilde{\mathbf{H}} ) \} ) \triangleq \mathbb{E} \{ \partial_{\bm{\lambda}}  \mathbf{w}^J(\tilde{\mathbf{h}}(\bm{\theta},\tilde{\mathbf{H}}), \bm{\lambda})  \allowbreak \partial_{\mathbf{w}} (R_k -  r_k(\bm{\theta},  \mathbf{w}^J ( \tilde{\mathbf{h}}(\bm{\theta}, \tilde{\mathbf{H}}), \bm{\lambda}),\tilde{\mathbf{H}} ) )  \}$. Then, we establish the following TTS primal-dual decomposition theorem for problem \eqref{TTS_problem_detail_equi}.
\newtheorem{theorem}{Theorem}
\begin{theorem}[TTS Primal-Dual Decomposition] \label{theorem_2}
	\emph{
	Suppose that for every $\tilde{\mathbf{H}} \in \Omega$, $\mathbf{w}^J(\tilde{\mathbf{h}}(\bm{\theta},\tilde{\mathbf{H}}), \bm{\lambda}) $ is a continuously differentiable function of $\bm{\theta} \in \mathcal{F}^N$ and $\bm{\lambda} \succeq \mathbf{0}$. Let $(\bm{\theta}^*, \bm{\lambda}^*)$ denote a stationary point of $\mathcal{P}_L^J$, i.e., there exist Lagrange multipliers $\tilde{\bm{\lambda}} = [\tilde{\lambda}_1,\cdots, \tilde{\lambda}_K]^T \succeq \mathbf{0}$ such that the following KKT conditions are satisfied:
	\begin{equation} \label{KKT_1}
	\begin{aligned}
	 \nabla_{\bm{\theta}}  \mathbb{E} \Bigg\{\sum\limits_{k\in \mathcal{K}} & \|\mathbf{w}_k^J(\tilde{\mathbf{h}}(\bm{\theta}^*,\tilde{\mathbf{H}}), \bm{\lambda}^*)  \|^2 \Bigg\}  + \sum\limits_{k \in \mathcal{K}} \tilde{\lambda}_k \nabla_{\bm{\theta}} (R_k \\
	& - \mathbb{E} \{ r_k(\bm{\theta}^*, \mathbf{w}^J(\tilde{\mathbf{h}}(\bm{\theta}^*,\tilde{\mathbf{H}}), \bm{\lambda}^*) ,\tilde{\mathbf{H}} ) \} ) =\mathbf{0},
	\end{aligned}
	\end{equation}
	\begin{equation} \label{KKT_2}
		\begin{aligned}
	 \nabla_{\bm{\lambda}} \mathbb{E} \Bigg\{\sum\limits_{k\in \mathcal{K}}&  \|\mathbf{w}_k^J(\tilde{\mathbf{h}}(\bm{\theta}^*,\tilde{\mathbf{H}}), \bm{\lambda}^*)  \|^2 \Bigg\} + \sum\limits_{k \in \mathcal{K}} \tilde{\lambda}_k \nabla_{\bm{\lambda}} (R_k \\
	& - \mathbb{E} \{ r_k(\bm{\theta}^*, \mathbf{w}^J(\tilde{\mathbf{h}}(\bm{\theta}^*,\tilde{\mathbf{H}}), \bm{\lambda}^*) ,\tilde{\mathbf{H}} ) \} ) =\mathbf{0},
	\end{aligned}
	\end{equation}
	\begin{equation} \label{KKT_3}
	\! \tilde{\lambda}_k (R_k \! -\! \mathbb{E} \{ r_k(\bm{\theta}^*, \mathbf{w}^J(\tilde{\mathbf{h}}(\bm{\theta}^*,\tilde{\mathbf{H}}), \bm{\lambda}^*) ,\tilde{\mathbf{H}} ) \} ) \!=\! 0, \;\forall k \! \in \! \mathcal{K},
	\end{equation}
	\begin{equation} \label{KKT_4}
	R_k - \mathbb{E} \{ r_k(\bm{\theta}^*, \mathbf{w}^J(\tilde{\mathbf{h}}(\bm{\theta}^*,\tilde{\mathbf{H}}), \bm{\lambda}^*) ,\tilde{\mathbf{H}} ) \}  \leq 0,\;\forall k \in \mathcal{K}.
	\end{equation}
	Then, the primal-dual pair $(\bm{\theta}^*, \bm{\varUpsilon}^* = \{ \mathbf{w}_k^J(\tilde{\mathbf{h}}(\bm{\theta}^*, \tilde{\mathbf{H}}),  \bm{\lambda}^*)  , \allowbreak \forall \tilde{\mathbf{H}} \},  \bm{\lambda}^*)$ satisfies the stationary conditions of problem \eqref{TTS_problem_detail_equi}, i.e., \eqref{stationary_1}, \eqref{stationary_2} and \eqref{stationary_3}, up to an error of $O(e^J(\bm{\theta}^*, \bm{\lambda}^*))$, where $\lim_{J\rightarrow \infty} e^J(\bm{\theta}^*, \bm{\lambda}^*) = 0$, providing that the following linear independence regularity condition (LIRC) holds: the gradients of the long-term constraints $\nabla_{\bm{\lambda}} (R_k - \mathbb{E} \{ r_k(\bm{\theta}^*, \mathbf{w}^J(\tilde{\mathbf{h}}(\bm{\theta}^*,\tilde{\mathbf{H}}), \bm{\lambda}^*) ,\tilde{\mathbf{H}} ) \} )$, $\forall k \in \mathcal{K}$ are linearly independent.
}
\end{theorem}
\begin{proof}
	Please refer to Appendix \ref{appendix_b}.
\end{proof}

Note that the LIRC in Theorem \ref{theorem_2} is used to guarantee that the complementary slackness condition in \eqref{stationary_3}  is satisfied. Besides, the condition that $\mathbf{w}^J(\tilde{\mathbf{h}}(\bm{\theta},\tilde{\mathbf{H}}), \bm{\lambda}) $ is continuously differentiable can be guaranteed by the short-term sub-algorithm, as will be elaborated in Section \ref{sec_short_term}. Theorem \ref{theorem_2} is promising since it shows that a stationary solution (up to the error of $e^J(\bm{\theta}^*, \bm{\lambda}^*)$) can be found by  solving the long-term problem $\mathcal{P}_L^J$ to obtain a stationary point $(\bm{\theta}^*, \bm{\lambda}^*)$, and finding a stationary point $\mathbf{w}^J(\tilde{\mathbf{h}}(\bm{\theta},\tilde{\mathbf{H}}), \bm{\lambda})$ of $\mathcal{P}_S(\tilde{\mathbf{h}}(\bm{\theta},\tilde{\mathbf{H}}), \bm{\lambda})$ for each $\tilde{\mathbf{H}}$. In the following, we will present how to solve $\mathcal{P}_S(\tilde{\mathbf{h}}(\bm{\theta},\tilde{\mathbf{H}}), \bm{\lambda} )$ and $\mathcal{P}_L^J$ efficiently.

\subsection{Solving the Short-term Problem} \label{sec_short_term}
With given $\bm{\theta}$ and $\bm{\lambda}$, we observe that $\mathcal{P}_S(\tilde{\mathbf{h}}(\bm{\theta},\tilde{\mathbf{H}}), \bm{\lambda} )$ is an unconstrained optimization problem whose objective function is a weighted-sum of the total transmit power $ \sum\nolimits_{k\in \mathcal{K}} \|\mathbf{w}_k\|^2$ and the achievable rates $\{ r_k(\bm{\theta},\{\mathbf{w}_k\}, \tilde{\mathbf{H}})\}$ with the weighting factors given by $\{-\lambda_k \}$.\footnote{The constant $\sum\nolimits_{k \in \mathcal{K}}\lambda_k R_k$ is omitted without loss of optimality.} Therefore, the WMMSE method \cite{Shi2011WMMSE} can be employed to iteratively obtain a stationary solution of $\mathcal{P}_S(\tilde{\mathbf{h}}(\bm{\theta},\tilde{\mathbf{H}}), \bm{\lambda} )$. Specifically, by introducing auxiliary receive (scaling) coefficients $\{u_k\}$ and weighting factors $\{q_k\}$\footnote{Note that these weighting factors are auxiliary variables that are different from $\{-\lambda_k \}$.}, $\mathcal{P}_S(\tilde{\mathbf{h}}(\bm{\theta},\tilde{\mathbf{H}}), \bm{\lambda} )$ is equivalent to the following problem:
\begin{equation} \label{WMMSE_problem}
\min \limits_{\{\mathbf{w}_k,\;u_k,\;q_k \}} \; \sum\limits_{k\in \mathcal{K}} \|\mathbf{w}_k\|^2+ \sum\limits_{k \in \mathcal{K}}\lambda_k (q_k e_k -\log(q_k)),
\end{equation}
where $e_k = \mathbb{E}\{ (u_k^H y_k - s_k)(u_k^H y_k-s_k)^H\} $ denotes the mean squared error (MSE) of user $k$ and is given by
\begin{equation}
e_k = |u_k^H  {\tilde{\mathbf{h}}_k^H} \mathbf{w}_k-1|^2 + \sum\limits_{j \in \mathcal{K} \backslash k} |u_k^H {\tilde{\mathbf{h}}_k^H} \mathbf{w}_j|^2 + \sigma_k^2 |u_k|^2.
\end{equation}

Then, we observe that problem \eqref{WMMSE_problem} can be efficiently solved by alternately optimizing one of the block variables $\{u_k \}$, $\{q_k \}$ and $\{\mathbf{w}_k\}$ with the others being fixed and the details are given as follows.
\subsubsection{Optimizing $u_k$} It can be seen that optimizing $\{u_k\}$ is equivalent to minimizing the sum-MSE $\sum\nolimits_{k \in \mathcal{K}} e_k$ (since $\lambda_k$'s are long-term variables and fixed during optimizing $u_k$), which leads to the following linear minimum MSE (MMSE) receive coefficients:
\begin{equation} \label{u_update}
u_k = \frac{\tilde{\mathbf{h}}_k^H \mathbf{w}_k}{\sum\limits_{j \in \mathcal{K}} |\tilde{\mathbf{h}}_k^H \mathbf{w}_j|^2 + \sigma_k^2}, \;\forall k \in \mathcal{K}.
\end{equation}
\subsubsection{Optimizing $q_k$} The optimal solution can be easily obtained as 
\begin{equation}  \label{q_update}
q_k = \frac{1}{e_k},\;\forall k \in \mathcal{K}.
\end{equation}
\subsubsection{Optimizing $\mathbf{w}_k$} By resorting to the first-order optimality condition of problem \eqref{WMMSE_problem} w.r.t. $\mathbf{w}_k$, we have
\begin{equation} \label{w_update}
\mathbf{w}_k \! = \! \lambda_k q_k u_k\Bigg(\mathbf{I} \!+\! \sum\limits_{k \in \mathcal{K}} \lambda_k q_k |u_k|^2 \tilde{\mathbf{h}}_k \tilde{\mathbf{h}}_k^H \Bigg)^{\!-\!1} \tilde{\mathbf{h}}_k,\;\forall k \! \in \! \mathcal{K}.
\end{equation}
Therefore, the short-term problem $\mathcal{P}_S(\tilde{\mathbf{h}}(\bm{\theta},\tilde{\mathbf{H}}), \bm{\lambda} )$ can be efficiently solved by iterating over the abovementioned three steps and the details are shown in Algorithm \ref{WMMSE_algorithm}.

In addition, we observe that: 1) for any $\bm{\theta} \in \mathcal{F}^N$, $\mathbf{w} =[\mathbf{w}_1^T, \cdots, \mathbf{w}_K^T]^T \in \mathbb{C}^{MK \times 1}$, $\bm{\lambda} \succeq \mathbf{0}$ and iteration number $j$, $\mathbf{w}^{j} (\tilde{\mathbf{h}}(\bm{\theta},\tilde{\mathbf{H}}), \bm{\lambda})$ is differentiable w.r.t. $\bm{\theta}$ and $\bm{\lambda}$ with probability one (w.p.1.); 2) for any $\bm{\theta} \in \mathcal{F}^N$  and $\bm{\lambda} \succeq \mathbf{0}$, the sequence $\{ \mathbf{w}^{j} (\tilde{\mathbf{h}}(\bm{\theta},\tilde{\mathbf{H}}), \bm{\lambda})\}$ converges to a stationary solution of $\mathcal{P}_S(\tilde{\mathbf{h}}(\bm{\theta},\tilde{\mathbf{H}}), \bm{\lambda} )$, w.p.1. \cite[Theorem 3]{Shi2011WMMSE}. On the other hand, we can always choose a proper initialization method for $\{\mathbf{w}_k\}$, such that it is differentiable w.r.t. $\bm{\theta}$ and $\bm{\lambda}$, w.p.1., therefore, the continuously differentiable condition in Theorem \ref{theorem_2} can be guaranteed.

\begin{algorithm}[t] \small
	\caption{{WMMSE Algorithm for Solving Problem \eqref{short_term_problem}}} \label{WMMSE_algorithm}
	\begin{algorithmic}
		\STATE \textbf{Initialize}: $\{\mathbf{w}_k\}$ and $i=1$.
		\REPEAT
		\STATE Update $\{ u_k\}$, $\{ q_k\}$ and $\{ \mathbf{w}_k\}$ successively according to \eqref{u_update}, \eqref{q_update} and \eqref{w_update}, respectively.
		\STATE $i \leftarrow i+1$. 
		\UNTIL{The maximum number of $J$ iterations is reached.}
	\end{algorithmic}
\end{algorithm}

\subsection{Solving the Long-term Problem} \label{sec_long_term}
In this subsection, we propose to solve the long-term problem $\mathcal{P}_L^J$ based on the CSSCA framework proposed in \cite{Liu_CSSCA_2019}, where an iterative algorithm is presented and in the each iteration, the long-term variables $\bm{\theta}$ and $\bm{\lambda}$ are updated by solving a convex surrogate problem obtained by replacing the objective and constraint functions in $\mathcal{P}_L^J$ with their convex surrogate functions. Specifically, each iteration consists of three steps, where the first step is to construct the surrogate functions, the second step is to solve the resulting convex optimization problem and the third step is to update the long-term variables $\bm{\theta}$ and $\bm{\lambda}$, as elaborated below. 
\subsubsection{Step 1} At the beginning of the $t$-th iteration, one random mini-batch $\{\tilde{\mathbf{H}}^t_j, {j=1,\cdots,B}\}$ of $B$ channel samples are generated according to the S-CSI and the surrogate functions, denoted by $\{\bar{f}_k^t(\bm{\theta}, \bm{\lambda}), k\in 0 \cup \mathcal{K}\}$, are constructed based on the mini-batch $\{\tilde{\mathbf{H}}^t_j\}$ and the current values $(\bm{\theta}^t,\bm{\lambda}^t)$, where $\bar{f}_0^t(\bm{\theta}, \bm{\lambda})$ is the surrogate function for the objective function $\mathbb{E} \{\sum\nolimits_{k\in \mathcal{K}} \|\mathbf{w}_k^J(\tilde{\mathbf{h}}(\bm{\theta},\tilde{\mathbf{H}}), \bm{\lambda})\|^2 \}$, while $\{\bar{f}_k^t(\bm{\theta}, \bm{\lambda}),k\in\mathcal{K} \}$ are the surrogate functions for the constraint functions $R_k - \mathbb{E} \{  r_k(\bm{\theta}, \{ \mathbf{w}_k^J(\tilde{\mathbf{h}}(\bm{\theta},\tilde{\mathbf{H}}), \bm{\lambda})  \},\tilde{\mathbf{H}} ) \}, \forall k \in \mathcal{K}$. Specifically, $\{\bar{f}_k^t(\bm{\theta}, \bm{\lambda})\}$'s can be constructed as \cite{Liu_CSSCA_2019}
\begin{equation} \label{surrogate_function_update}
\begin{aligned}
\bar{f}_k^t(\bm{\theta}, \bm{\lambda}) \!  = & f_k^t \!+\! (\mathbf{f}_{\bm{\theta},k}^t \!+\! \mathbf{f}_{\mathbf{w},k}^t)^T(\bm{\theta} \!-\! \bm{\theta}^t) \!+\! (\mathbf{f}_{\bm{\lambda},k}^t)^T(\bm{\lambda} \!-\! \bm{\lambda}^t) \\
& \!+\! \tau_k (\|\bm{\theta} \!-\! \bm{\theta}^t\|^2 \!+\! \|\bm{\lambda} \!-\! \bm{\lambda}^t\|^2),\;\forall k \in 0 \cup \mathcal{K},
\end{aligned}
\end{equation}
where $\tau_k > 0$ can be any constant and is imposed to ensure that $\bar{f}_k^t(\bm{\theta}, \bm{\lambda}) $ is uniformly strongly convex in $\bm{\theta}$ and $\bm{\lambda}$; $f_0^t$ and $f_k^{t}, \forall k  \in \mathcal{K}$ are approximations for $\mathbb{E} \{\sum\nolimits_{k\in \mathcal{K}} \|\mathbf{w}_k^J(\tilde{\mathbf{h}}(\bm{\theta}^t,\tilde{\mathbf{H}}), \bm{\lambda}^t)\|^2 \}$ and $R_k - \mathbb{E} \{  r_k(\bm{\theta}^t,\{ \mathbf{w}_k^J(\tilde{\mathbf{h}}(\bm{\theta}^t,\tilde{\mathbf{H}}), \bm{\lambda}^t)  \},\tilde{\mathbf{H}} ) \}$, respectively, and they are updated iteratively according to
\begin{equation}
f_0^{t} = (1-\rho^t )f_0^{t-1}+ \rho^t \frac{1}{B} \sum\limits_{j=1}^B \|\mathbf{w}_j^t\|^2,
\end{equation}
\begin{equation} 
\begin{aligned}
f_k^{t} = (1 & -\rho^t )f_k^{t-1} \\
& + \rho^t \frac{1}{B} \sum\limits_{j=1}^B (R_k -  r_k(\bm{\theta}^t, \mathbf{w}_j^t ,{\tilde{\mathbf{H}}_j^t} ) ),\;\forall k \in \mathcal{K},
\end{aligned}
\end{equation}
where $\{\rho^t \in (0,1]\} $ is a decreasing sequence satisfying $\rho^t \rightarrow 0$, $\sum_t \rho^t = \infty$ and $\sum_t (\rho^t)^2 < \infty$; $f_k^{-1} = 0$ and $\mathbf{w}_j^t$ denotes the abbreviation for $\{ \mathbf{w}_k^J(\tilde{\mathbf{h}}(\bm{\theta}^t,\tilde{\mathbf{H}}_j^t), \bm{\lambda}^t )  \}$; $\mathbf{f}_{\mathbf{w},0}^t$ and $\mathbf{f}_{\bm{\lambda},0}^t$ are approximations for the gradients $\mathbb{E} \{ \partial_{\bm{\theta}} \mathbf{w}^J(\tilde{\mathbf{h}}(\bm{\theta}^t,\tilde{\mathbf{H}}), \bm{\lambda}^t) \partial_{\mathbf{w}} (\sum\nolimits_{k\in \mathcal{K}} \|\mathbf{w}_k^J(\tilde{\mathbf{h}}(\bm{\theta}^t,\tilde{\mathbf{H}}), \bm{\lambda}^t)\|^2)\} $ and $\mathbb{E} \{ \partial_{\bm{\lambda}} \mathbf{w}^J(\tilde{\mathbf{h}}(\bm{\theta}^t,\tilde{\mathbf{H}}), \bm{\lambda}^t) \partial_{\mathbf{w}}(\sum\nolimits_{k\in \mathcal{K}} \|\mathbf{w}_k^J(\tilde{\mathbf{h}}(\bm{\theta}^t,\tilde{\mathbf{H}}), \bm{\lambda}^t)\|^2 ) \} $, respectively, and $ \mathbf{f}_{\bm{\theta},0}^t = \mathbf{0}$; $\mathbf{f}_{\bm{\theta},k}^t$, $ \mathbf{f}_{\mathbf{w},k}^t$ and $\mathbf{f}_{\bm{\lambda},k}^t$, $\forall k \in \mathcal{K}$ are approximations for the gradients $\mathbb{E} \{ \partial_{\bm{\theta}}(R_k  - r_k(\bm{\theta}^t,\{ \mathbf{w}_k^J(\tilde{\mathbf{h}}(\bm{\theta}^t, \allowbreak \tilde{\mathbf{H}}), \bm{\lambda}^t)  \},\tilde{\mathbf{H}} ) )\} $, $\mathbb{E} \{\partial_{\bm{\theta}} \mathbf{w}^J(\tilde{\mathbf{h}}(\bm{\theta}^t,\tilde{\mathbf{H}}), \bm{\lambda}^t) \partial_{\mathbf{w}}(R_k - r_k(\bm{\theta}^t, \{ \mathbf{w}_k^J(\allowbreak  \tilde{\mathbf{h}}(\bm{\theta}^t, \tilde{\mathbf{H}}), \bm{\lambda}^t)  \},\tilde{\mathbf{H}} ) )\} $ and $\mathbb{E} \{ \nabla_{\bm{\lambda}}(R_k - r_k(\bm{\theta}^t,\{ \mathbf{w}_k^J(\tilde{\mathbf{h}}(\bm{\theta}^t, \tilde{\mathbf{H}}), \allowbreak \bm{\lambda}^t)  \},\tilde{\mathbf{H}} )) \} $, respectively. $\mathbf{f}_{\mathbf{w},0}^t$ and $\mathbf{f}_{\bm{\lambda},0}^t$ are updated iteratively according to
\begin{equation}
\begin{aligned}
\mathbf{f}_{\mathbf{w},0}^t = & (1-\rho^t)\mathbf{f}_{\mathbf{w},0}^{t-1} + \rho^t \frac{1}{B} \sum\limits_{j=1}^B \partial_{\bm{\theta}}\mathbf{w}^J(\tilde{\mathbf{h}}(\bm{\theta}^t,\tilde{\mathbf{H}}_j^t), \bm{\lambda}^t) \\
& \times \partial_{\mathbf{w}} \left(  \sum\limits_{k\in \mathcal{K}} \|\mathbf{w}_k^J(\tilde{\mathbf{h}}(\bm{\theta}^t,\tilde{\mathbf{H}}_j^t), \bm{\lambda}^t)\|^2\right),
\end{aligned}
\end{equation}
\begin{equation}
 \begin{aligned}
\mathbf{f}_{\bm{\lambda},0}^t= & (1-\rho^t)\mathbf{f}_{\bm{\lambda},0}^{t-1}  + \rho^t \frac{1}{B} \sum\limits_{j=1}^B  \partial_{\bm{\lambda}}\mathbf{w}^J(\tilde{\mathbf{h}}(\bm{\theta}^t,\tilde{\mathbf{H}}_j^t), \bm{\lambda}^t) \\
& \times  \partial_{\mathbf{w}} \left(  \sum\limits_{k\in \mathcal{K}} \|\mathbf{w}_k^J(\tilde{\mathbf{h}}(\bm{\theta}^t,\tilde{\mathbf{H}}_j^t), \bm{\lambda}^t)\|^2\right), \\
\end{aligned}
\end{equation}
while $\mathbf{f}_{\bm{\theta},k}^t$, $ \mathbf{f}_{\mathbf{w},k}^t$ and $\mathbf{f}_{\bm{\lambda},k}^t$, $\forall k \in \mathcal{K}$ are updated by
\begin{equation}
\begin{aligned}
\mathbf{f}_{\bm{\theta},k}^t \!=\! (1\!-\! \rho^t)\mathbf{f}_{\bm{\theta},k}^{t-1} \!+\! \rho^t \frac{1}{B} \sum\limits_{j=1}^B \partial_{\bm{\theta}}(R_k \!-\! r_k(\bm{\theta}^t, \mathbf{w}_j^t,{\tilde{\mathbf{H}}_j^t} ) ),
\end{aligned}
\end{equation}
\begin{equation}
\begin{aligned}
\mathbf{f}_{\mathbf{w},k}^t = & (1-\rho^t)\mathbf{f}_{\mathbf{w},k}^{t-1}  + \rho^t \frac{1}{B} \sum\limits_{j=1}^B \partial_{\bm{\theta}} \mathbf{w}^J(\tilde{\mathbf{h}}(\bm{\theta}^t,\tilde{\mathbf{H}}_j^t), \bm{\lambda}^t) \\
& \times  \partial_{\mathbf{w}}(R_k - r_k(\bm{\theta}^t,\mathbf{w}_j^t,{\tilde{\mathbf{H}}_j^t} ) ), \\
 \end{aligned}
 \end{equation}
 \begin{equation}
 \begin{aligned}
\mathbf{f}_{\bm{\lambda},k}^t = & (1-\rho^t)\mathbf{f}_{\bm{\lambda},k}^{t-1}  + \rho^t \frac{1}{B} \sum\limits_{j=1}^B  \partial_{\bm{\lambda}} \mathbf{w}^J(\tilde{\mathbf{h}}(\bm{\theta}^t,\tilde{\mathbf{H}}_j^t), \bm{\lambda}^t) \\
& \times \partial_{\mathbf{w}}(R_k - r_k(\bm{\theta}^t, \mathbf{w}_j^t ,{\tilde{\mathbf{H}}_j^t} ) ), \\
\end{aligned}
\end{equation}
with $\mathbf{f}_{\bm{\theta},k}^{-1} = \mathbf{0}$, $\mathbf{f}_{\mathbf{w},k}^{-1} = \mathbf{0}$ and $\mathbf{f}_{\bm{\lambda},k}^{-1} = \mathbf{0}$, $\forall k \in 0 \cup \mathcal{K}$. Note that the approximations $f_0^{t}$, $\mathbf{f}_{\mathbf{w},0}^t$ and $\mathbf{f}_{\bm{\lambda},0}^t$ can converge to the true values of $\mathbb{E} \{\sum\nolimits_{k\in \mathcal{K}} \|\mathbf{w}_k^J(\tilde{\mathbf{h}}(\bm{\theta}^t,\tilde{\mathbf{H}}), \bm{\lambda}^t)  \|^2 \}$ and its gradients w.r.t. $\bm{\theta}$ and $\bm{\lambda}$ (i.e., $\nabla_{\bm{\theta}}  \mathbb{E} \{\sum\nolimits_{k\in \mathcal{K}} \|\mathbf{w}_k^J(\tilde{\mathbf{h}}(\bm{\theta}^t, \tilde{\mathbf{H}}), \allowbreak\bm{\lambda}^t)  \|^2 \} $ and $\nabla_{\bm{\lambda}}  \mathbb{E} \{\sum\nolimits_{k\in \mathcal{K}} \|\mathbf{w}_k^J(\tilde{\mathbf{h}}(\bm{\theta}^t,\tilde{\mathbf{H}}), \bm{\lambda}^t)  \|^2 \} $), as $t \rightarrow \infty$, 
while $f_k^{t}$, $\mathbf{f}_{\bm{\theta},k}^t \allowbreak + \mathbf{f}_{\mathbf{w},k}^t$ and $\mathbf{f}_{\bm{\lambda},k}^t$, $\forall k \in \mathcal{K}$ can converge to the true values of $R_k - \mathbb{E} \{ r_k(\bm{\theta}^t,\{ \mathbf{w}_k^J(\tilde{\mathbf{h}}(\bm{\theta}^t,\tilde{\mathbf{H}}), \bm{\lambda}^t)  \},\tilde{\mathbf{H}} ) \} $ and its gradients w.r.t. $\bm{\theta}$ and $\bm{\lambda}$ (i.e., $\nabla_{\bm{\theta}}(R_k - \mathbb{E} \{ r_k(\bm{\theta}^t,\{ \mathbf{w}_k^J(\tilde{\mathbf{h}}(\bm{\theta}^t,\tilde{\mathbf{H}}), \bm{\lambda}^t)  \},\allowbreak \tilde{\mathbf{H}} ) \} )$ and $\nabla_{\bm{\lambda}} (R_k - \mathbb{E} \{ r_k(\bm{\theta}^t,\{ \mathbf{w}_k^J(\tilde{\mathbf{h}}(\bm{\theta}^t,\tilde{\mathbf{H}}), \bm{\lambda}^t)  \},\tilde{\mathbf{H}} ) \} )$), as $t \rightarrow \infty$, which will be proved in Lemma \ref{lemma_1} of Section \ref{Section_proof}. Therefore, the function values and gradients of $\bar{f}_k^t(\bm{\theta}, \bm{\lambda}) $, $\forall k \in 0 \cup \mathcal{K}$  are consistent with those of the objective and constraint functions $\mathbb{E} \{\sum\nolimits_{k\in \mathcal{K}} \|\mathbf{w}_k^J(\tilde{\mathbf{h}}(\bm{\theta},\tilde{\mathbf{H}}), \bm{\lambda})  \|^2 \}$ and $R_k - \mathbb{E} \{ r_k(\bm{\theta},\{ \mathbf{w}_k^J(\tilde{\mathbf{h}}(\bm{\theta},\tilde{\mathbf{H}}), \bm{\lambda})  \},\tilde{\mathbf{H}} ) \} $, $\forall k \in \mathcal{K}$ at the current values $( \bm{\theta}^t, \bm{\lambda}^t)$, which is the key to guarantee the convergence of the proposed algorithm to a stationary point.
\subsubsection{Step 2} In this step, the optimal solution $(\bar{\bm{\theta}}^t,\bar{\bm{\lambda}}^t )$ of the following problem is obtained:
\begin{equation} \label{long_term_problem_appro1}
\begin{aligned}
(\bar{\bm{\theta}}^t,\bar{\bm{\lambda}}^t) = \arg\min\limits_{ \bm{\theta} \in \mathcal{F}^N,\;\bm{\lambda} \succeq \mathbf{0} } \; & \bar{f}_0^t(\bm{\theta}, \bm{\lambda})\\
\textrm{s.t.}\; & \bar{f}_k^t(\bm{\theta}, \bm{\lambda}) \leq 0,\forall k \in \mathcal{K},
\end{aligned}
\end{equation}
which is a convex approximation of $\mathcal{P}_L^J$. It is worth noting that problem \eqref{long_term_problem_appro1} is not necessarily feasible, thus the optimal solution $(\bar{\bm{\theta}}^t,\bar{\bm{\lambda}}^t )$ of the following convex problem is attained if problem \eqref{long_term_problem_appro1} turns out to be infeasible:
\begin{equation} \label{long_term_problem_appro}
\begin{aligned}
(\bar{\bm{\theta}}^t,\bar{\bm{\lambda}}^t) = \arg\min\limits_{ \bm{\theta} \in \mathcal{F}^N,\;\bm{\lambda} \succeq \mathbf{0},\; \alpha} \; & \alpha\\
\textrm{s.t.}\; & \bar{f}_k^t(\bm{\theta}, \bm{\lambda}) \leq \alpha,\forall k \in \mathcal{K}.
\end{aligned}
\end{equation}
Note that optimizing $\bm{\theta}$ and $\bm{\lambda}$ in problem \eqref{long_term_problem_appro1} is able to minimize the average transmit power since $\mathbf{w}^J(\tilde{\mathbf{h}}(\bm{\theta},\tilde{\mathbf{H}}), \bm{\lambda})$ is a function of $\bm{\theta}$ and $\bm{\lambda}$, while by solving problem \eqref{long_term_problem_appro}, we approximately minimize the gap between $\mathbb{E} \{r_k(\bm{\theta},\{ \mathbf{w}_k^J(\tilde{\mathbf{h}}(\bm{\theta},\tilde{\mathbf{H}}), \bm{\lambda})  \},\tilde{\mathbf{H}} ) \} $ and the corresponding rate target $R_k$, i.e., $\alpha$, which helps to pull the solution to the feasible region when the current problem at iteration $t$ is infeasible.
It can be seen that both problems \eqref{long_term_problem_appro1} and \eqref{long_term_problem_appro} are convex with simple quadratic objective and constraint functions, which can be efficiently solved by off-the-shelf solvers, such as CVX \cite{CVX}.

\subsubsection{Step 3} After obtaining $(\bar{\bm{\theta}}^t,\bar{\bm{\lambda}}^t)$, $\bm{\theta}$  and $\bm{\lambda}$ are updated according to
\begin{equation} \label{long_term_update}
\begin{aligned}
\bm{\theta}^{t+1} = (1-\gamma^t) \bm{\theta}^{t} + \gamma^t \bar{\bm{\theta}}^t,\\
\bm{\lambda}^{t+1} = (1-\gamma^t) \bm{\lambda}^{t} + \gamma^t \bar{\bm{\lambda}}^t,
\end{aligned}
\end{equation}
where $\{\gamma^t \in (0,1] \}$ is a decreasing sequence satisfying $\gamma^t \rightarrow 0$, $\sum_t \gamma^t =\infty$, $\sum_t (\gamma^t)^2 < \infty$ and $\lim_{t \rightarrow \infty} \gamma^t/\rho^t = 0$. Note that since $\gamma^t$ satisfies $\gamma^t \in (0,1]$, every point in $\{\bm{\theta}^t \}_{t=1}^{\infty}$ will automatically lie in the feasible region of the reflection phase shifts, i.e.,  $\mathcal{F}$, according to the updating rule in \eqref{long_term_update} and the fact that $\bar{\bm{\theta}}^t \in \mathcal{F}^N$.

\subsection{Deep Unfolding for Gradient Information Extraction}
When solving the long-term problem $\mathcal{P}_L^J$, we need to calculate the gradients of $ \mathbf{w}^J(\tilde{\mathbf{h}}(\bm{\theta},\tilde{\mathbf{H}}), \bm{\lambda})$ w.r.t. the long-term variables $\bm{\theta}$ and $\bm{\lambda}$, where $ \mathbf{w}^J(\tilde{\mathbf{h}}(\bm{\theta},\tilde{\mathbf{H}}), \bm{\lambda})$ is obtained by running the short-term sub-algorithm (i.e., Algorithm \ref{WMMSE_algorithm}) for $J$ iterations. Since obtaining $ \mathbf{w}^J(\tilde{\mathbf{h}}(\bm{\theta},\tilde{\mathbf{H}}), \bm{\lambda})$ involves an iterative algorithm, it is difficult to calculate its gradients in closed-form, which hinders the implementation of the long-term sub-algorithm. To tackle this difficulty, we propose to employ the deep unfolding technique to establish a neural network (NN) representation of $\mathbf{w}^J(\tilde{\mathbf{h}}(\bm{\theta},\tilde{\mathbf{H}}), \bm{\lambda})$ (or equivalently, Algorithm \ref{WMMSE_algorithm}), then we can calculate the gradients $\partial_{\bm{\theta}} \mathbf{w}^J(\tilde{\mathbf{h}}(\bm{\theta},\tilde{\mathbf{H}}), \bm{\lambda})$ and $\partial_{\bm{\lambda}} \mathbf{w}^J(\tilde{\mathbf{h}}(\bm{\theta},\tilde{\mathbf{H}}), \bm{\lambda})$ using the well-known back-propagation (BP) method \cite{rumelhart1986learning}. The BP method is essentially a gradient descent type algorithm, which requires computing the gradient of the loss function (i.e., $\mathbf{w}^J(\tilde{\mathbf{h}}(\bm{\theta},\tilde{\mathbf{H}}), \bm{\lambda})$ in our case) with respect to each weight of the NN by the chain rule, computing the gradient one layer at a time and iterating backward from the last layer to the first layer. Please refer to \cite{rumelhart1986learning} for more details.

Specifically, based on the WMMSE algorithm described in Section \ref{sec_short_term}, we construct the WMMSE network by unfolding the iterations of \eqref{u_update}-\eqref{w_update}, as shown in Fig. \ref{fig_WMMSE_network}. The WMMSE network is defined over a data flow graph, which consists of $J$ layers and each layer  corresponds to one iteration in Algorithm \ref{WMMSE_algorithm}; each layer further contains three sub-layers consisting of the three updating steps in \eqref{u_update}-\eqref{w_update}. The nodes in the graph correspond to different operations in Algorithm \ref{WMMSE_algorithm}, the directed edges represent the data flows between these operations and the long-term variables $\bm{\theta}$ and $\bm{\lambda}$ are viewed as network parameters. Besides, since the purpose of the WMMSE network is to help calculate the gradients $\partial_{\bm{\theta}} \mathbf{w}^J(\tilde{\mathbf{h}}(\bm{\theta},\tilde{\mathbf{H}}), \bm{\lambda})$ and $\partial_{\bm{\lambda}} \mathbf{w}^J(\tilde{\mathbf{h}}(\bm{\theta},\tilde{\mathbf{H}}), \bm{\lambda})$ for the optimization of the long-term variables $(\bm{\theta}, \bm{\lambda})$, it is unnecessary to include activation functions in the WMMSE network \cite{bishop2006pattern}.

\begin{figure*}[t] 
	\setlength{\abovecaptionskip}{-0cm} 
	\setlength{\belowcaptionskip}{-0cm} 
	\centering
	\scalebox{0.5}{\includegraphics{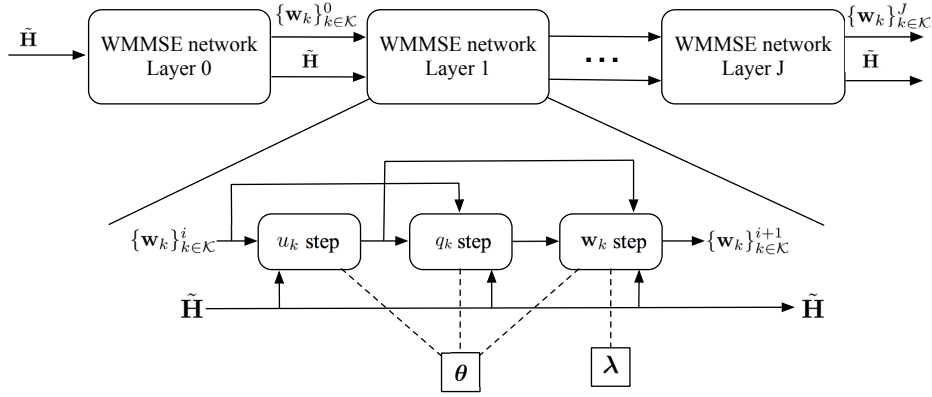}} 
	\caption{Structure of the WMMSE network.}
	\label{fig_WMMSE_network}
\end{figure*}

\subsection{Overall Algorithm with Convergence/Complexity Analysis} \label{Section_proof}
To summarize, the long-term IRS phase shifts and Lagrange multipliers $(\bm{\theta}, \bm{\lambda})$ can be optimized in an offline manner based on the statistics of the channel $\tilde{\mathbf{H}}$, or a data set containing a large number of channel samples. Once the optimized $(\bm{\theta}, \bm{\lambda})$ is obtained, the short-term active beamforming vectors $\{\mathbf{w}_k\}$ for each channel realization can be calculated in an online manner based on the effective channels $\{\tilde{\mathbf{h}}_k \}$ using Algorithm \ref{WMMSE_algorithm}. The proposed PDD-TJAPB algorithm is given in Algorithm \ref{PDD_CSSCA_algorithm} and the corresponding flow chart is shown in Fig. \ref{flow_chart_PDD_TJAPB}. It is important to mention that our proposed algorithm only requires the S-CSI and effective I-CSI to solve problem \eqref{TTS_problem_detail_equi} (i.e., using the S-CSI to generate channel samples $\{\tilde{\mathbf{H}}_j^t \}$ for long-term optimization and the effective I-CSI to design $\{\mathbf{w}_k \}$ for short-term optimization), despite that the full I-CSI is given in \eqref{TTS_problem_detail_equi}. This is practically appealing because obtaining the real-time full CSI will cause unacceptable channel estimation overhead in the considered IRS-aided communication system, especially for large $N$.

\begin{figure*}[t] 
	\setlength{\abovecaptionskip}{-0cm} 
	\setlength{\belowcaptionskip}{-0cm} 
	\centering
	\scalebox{0.42}{\includegraphics{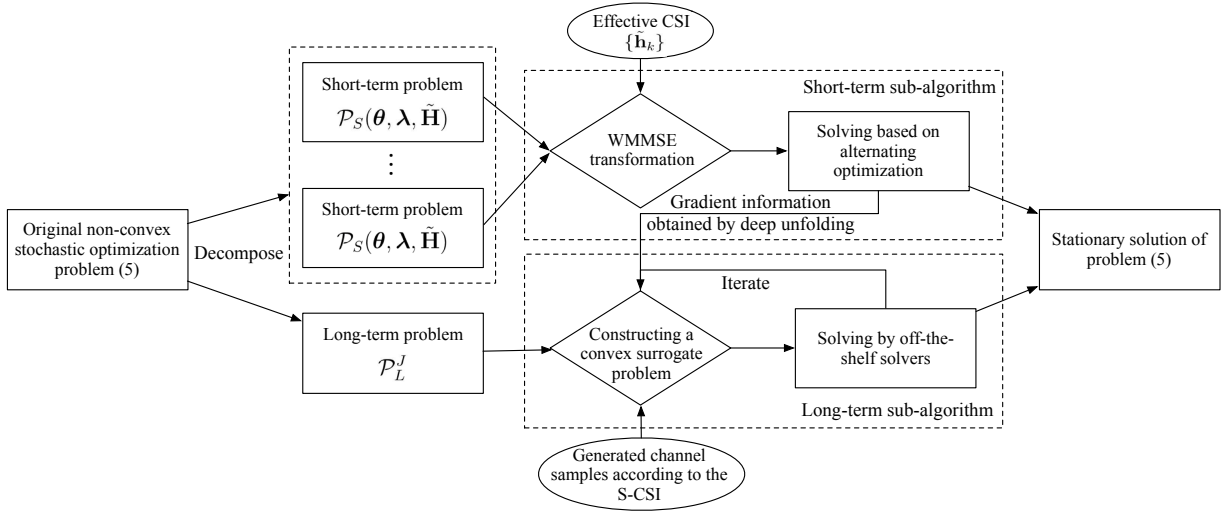}} 
	\caption{Flow chart of the proposed PDD-TJAPB algorithm.}
	\label{flow_chart_PDD_TJAPB}
\end{figure*}

\begin{algorithm}[t] \small
	\caption{{Proposed PDD-TJAPB Algorithm for Solving Problem \eqref{TTS_problem_detail_equi}}} \label{PDD_CSSCA_algorithm}
	\begin{algorithmic}
		\STATE \textbf{Input}: $\{\rho^t\}$, $\{\gamma^t\}$ and $B$. \textbf{Initialize}: $\bm{\theta}^0$, $\bm{\lambda}^0$ and $t=0$. 
		\STATE \textbf{Long-term}: (Passive IRS phase-shifts design):
		\STATE \begin{itemize}
			\item \textbf{Step 1}: Generate a mini-batch $\{\tilde{\mathbf{H}}^t_j\}$ of $B$ channel samples according to the known S-CSI and update the surrogate functions $\bar{f}_k^t(\bm{\theta}, \bm{\lambda})$, $\forall k \in \mathcal{K}$, by \eqref{surrogate_function_update}, where $\{\mathbf{w}_k^J(\tilde{\mathbf{h}}(\bm{\theta}^t,\tilde{\mathbf{H}}_j^t), \bm{\lambda}^t)\}$ are obtained by employing Algorithm \ref{WMMSE_algorithm} to solve problem \eqref{WMMSE_problem} with given generated channel samples and fixed $(\bm{\theta}^{t}, \bm{\lambda}^{t})$.
			\item \textbf{Step 2}: \textbf{If} problem \eqref{long_term_problem_appro1} is feasible, \textbf{then} solve problem \eqref{long_term_problem_appro1} to obtain the optimal $(\bar{\bm{\theta}}^t, \bar{\bm{\lambda}}^{t})$, \textbf{else} solve problem \eqref{long_term_problem_appro} to obtain $(\bar{\bm{\theta}}^t, \bar{\bm{\lambda}}^{t})$, \textbf{end if}.
			\item \textbf{Step 3}: Update $(\bm{\theta}^{t+1}, \bm{\lambda}^{t+1})$ according to \eqref{long_term_update}.
			\item Let $t = t + 1$ and return to \textbf{Step 1}. Repeat the above until convergence.
		\end{itemize}
		\STATE \textbf{Short-term}: (Active beamforming design):
		\STATE \begin{itemize}
			\item Apply Algorithm \ref{WMMSE_algorithm} to solve problem \eqref{WMMSE_problem} with given $\bm{\lambda}$ and $\{\tilde{\mathbf{h}}_k\}$ to obtain $\{\mathbf{w}_k\}$ for each time slot.
		\end{itemize}
	\end{algorithmic}
\end{algorithm}

To show the convergence of Algorithm \ref{PDD_CSSCA_algorithm}, we first prove an essential lemma that establishes several important properties of the surrogate functions $\bar{f}_k^t(\bm{\theta}, \bm{\lambda})$'s, which will be used in the sequel.
\newtheorem{lemma}{Lemma}
\begin{lemma}[Properties of $\bar{f}_k^t(\bm{\theta}, \bm{\lambda})$'s] \label{lemma_1}
	\emph{
	For all $k \in 0 \cup \mathcal{K}$ and $t=1,2,\cdots$, we have}
	
	\emph{
	1) $\bar{f}_k^t(\bm{\theta}, \bm{\lambda})$ is uniformly strongly convex in $\bm{\theta}$ and $\bm{\lambda}$.}

\emph{
2) $\bar{f}_k^t(\bm{\theta}, \bm{\lambda})$ is a Lipschitz continuous function w.r.t. $\bm{\theta}$ and $\bm{\lambda}$. Moreover, $\lim \sup_{t_1, t_2 \rightarrow \infty} |\bar{f}_k^{t_1}(\bm{\theta}, \bm{\lambda}) - \bar{f}_k^{t_2}(\bm{\theta}, \bm{\lambda})| - C\sqrt{\|\bm{\theta}^{t_1} - \bm{\theta}^{t_2}\|^2 + \| \bm{\lambda}^{t_1} -  \bm{\lambda}^{t_2}\|^2 } \leq 0$, $\forall \bm{\theta} \in \mathcal{F}^N$, $\bm{\lambda} \succeq \mathbf{0}$ for some constant $C > 0$. }

\emph{
	3) For any $ \bm{\theta} \in \mathcal{F}^N$ and $\bm{\lambda} \succeq \mathbf{0}$, the function $\bar{f}_k^t(\bm{\theta}, \bm{\lambda})$, its derivative, and its second order derivative are uniformly bounded.}

\emph{
	4) $\lim_{t\rightarrow \infty } |\bar{f}_0^t(\bm{\theta}^t, \bm{\lambda}^t) - \mathbb{E} \{\sum\nolimits_{k\in \mathcal{K}} \|\mathbf{w}_k^J(\tilde{\mathbf{h}}(\bm{\theta}^t, \tilde{\mathbf{H}}),  \bm{\lambda}^t)  \|^2  \}| \allowbreak = 0$, $\lim_{t\rightarrow \infty } \| \mathbf{f}_{\mathbf{w},0}^t - \nabla_{\bm{\theta}}  \mathbb{E} \{\sum\nolimits_{k\in \mathcal{K}} \|\mathbf{w}_k^J(\tilde{\mathbf{h}}(\bm{\theta}^t, \tilde{\mathbf{H}}),  \bm{\lambda}^t)  \|^2  \}  \| \allowbreak = 0$ and $\lim_{t\rightarrow \infty } \|\mathbf{f}_{\bm{\lambda},0}^t - \nabla_{\bm{\lambda}}  \mathbb{E} \{\sum\nolimits_{k\in \mathcal{K}} \|\mathbf{w}_k^J(\tilde{\mathbf{h}}(\bm{\theta}^t,\tilde{\mathbf{H}}), \allowbreak \bm{\lambda}^t)  \|^2  \} \| = 0$, while for $k \in \mathcal{K}$, $\lim_{t\rightarrow \infty } |\bar{f}_k^t(\bm{\theta}^t, \bm{\lambda}^t) - (R_k - \mathbb{E} \{ r_k(\bm{\theta}^t,\{ \mathbf{w}_k^J(\tilde{\mathbf{h}}(\bm{\theta}^t,\tilde{\mathbf{H}}), \bm{\lambda}^t)  \},\tilde{\mathbf{H}} ) \})| = 0$, $\lim_{t\rightarrow \infty } \|\mathbf{f}_{\bm{\theta},k}^t + \mathbf{f}_{\mathbf{w},k}^t - \nabla_{\bm{\theta}}(R_k - \mathbb{E} \{ r_k(\bm{\theta}^t,\{ \mathbf{w}_k^J(\tilde{\mathbf{h}}(\bm{\theta}^t,\tilde{\mathbf{H}}), \bm{\lambda}^t)  \},\tilde{\mathbf{H}} ) \} ) \| = 0$ and $\lim_{t\rightarrow \infty } \|\mathbf{f}_{\bm{\lambda},k}^t - \nabla_{\bm{\lambda}}(R_k - \mathbb{E} \{ r_k(\bm{\theta}^t,\{ \mathbf{w}_k^J(\tilde{\mathbf{h}}(\bm{\theta}^t, \tilde{\mathbf{H}}),  \bm{\lambda}^t)  \},\allowbreak \tilde{\mathbf{H}} ) \} ) \| = 0$.}

\emph{
5) Consider a subsequence $\{\bm{\theta}^{t_j}, \bm{\lambda}^{t_j} \}_{j=1}^{\infty}$ converging to a limit point $(\bm{\theta}^*, \bm{\lambda^*})$. There exist uniformly differentiable
functions $\{\hat{f}_k(\bm{\theta}, \bm{\lambda})\}$ such that
\begin{equation}
\lim_{j \rightarrow \infty} \bar{f}_k^{t_j}(\bm{\theta}, \bm{\lambda}) = \hat{f}_k(\bm{\theta}, \bm{\lambda}), \;\forall \bm{\theta} \in \mathcal{F}^N, \bm{\lambda} \succeq \mathbf{0}, 
\end{equation}
almost surely.
}
\end{lemma}
\begin{proof}
	Please refer to Appendix \ref{appendix_c}. 
\end{proof}
Besides, we define the following Slater's condition for $\mathcal{P}_L^J$. 
\begin{definition}[Slater's condition for $\mathcal{P}_L^J$] \label{slater_condition}
	\emph{
Given a subsequence $\{\bm{\theta}^{t_j}, \bm{\lambda}^{t_j} \}_{j=1}^{\infty}$ converging to a limit point $(\bm{\theta}^*, \bm{\lambda^*})$ and let $\hat{f}_k(\bm{\theta}, \bm{\lambda})$'s be the converged surrogate functions as defined in Lemma \ref{lemma_1} property 5. We say that the Slater's condition is satisfied at $(\bm{\theta}^*, \bm{\lambda}^*)$ if there exists $(\bm{\theta}^*, \bm{\lambda}^*) \in \textrm{relint}( \mathcal{F}^N \times \mathbb{R}^{K++})$ such that
\begin{equation}
\hat{f}_k(\bm{\theta}^*, \bm{\lambda}^*) <0, \;\forall k \in \mathcal{K}.
\end{equation}
}
\end{definition}
With Lemma \ref{lemma_1} and Definition \ref{slater_condition}, we have the following main result. 
\begin{theorem}[Convergence of Algorithm \ref{PDD_CSSCA_algorithm}] \label{theorem_3}
	\emph{
Suppose that the initial point $(\bm{\theta}^0 \in \mathcal{F}^N, \bm{\lambda}^0 \in \mathbb{R}^{K++})$ is a feasible point of $\mathcal{P}_L^J$, i.e., $\min_{k\in \mathcal{K}}  \mathbb{E} \{  r_k(\bm{\theta}^0,\{ \mathbf{w}_k^J(\tilde{\mathbf{h}}(\bm{\theta}^0,\tilde{\mathbf{H}}), \bm{\lambda}^0)  \},\tilde{\mathbf{H}} ) \} \geq R_k$.  Let $\{\bm{\theta}^t, \bm{\lambda}^t \}_{t=1}^{\infty}$ denote the values iteratively generated by Algorithm \ref{PDD_CSSCA_algorithm} with a sufficiently small initial step size $\gamma^0$, then every limit point $(\bm{\theta}^*, \bm{\lambda}^*)$ of $\{\bm{\theta}^t, \bm{\lambda}^t \}_{t=1}^{\infty}$ satisfying the LIRC in Theorem \ref{theorem_2} and the Slater's condition, almost surely satisfies the stationary conditions in \eqref{stationary_1}, \eqref{stationary_2} and \eqref{stationary_3} up to an error of $O(e^J(\bm{\theta}^*, \bm{\lambda}^*))$, which diminishes to zero as $\lim_{J \rightarrow \infty}e^J(\bm{\theta}^*, \bm{\lambda}^*) = 0$.
}
\end{theorem}
\begin{proof}
	From Lemma \ref{lemma_1} and the convergence analysis of the CSSCA framework in \cite{Liu_CSSCA_2019}, we can see that starting from a feasible initial point, every limit point $(\bm{\theta}^*, \bm{\lambda}^*)$ generated by the long-term optimization in Algorithm \ref{PDD_CSSCA_algorithm} is a stationary point of the long-term subproblem $\mathcal{P}_L^J$ almost surely, providing that the initial step size $\gamma^0$ is sufficiently small, and the Slater's condition is satisfied for $(\bm{\theta}^*, \bm{\lambda}^*)$. Then, it follows from Theorem \ref{theorem_2} that the solution $(\bm{\theta}^*,\bm{\varUpsilon}^* = \{ \mathbf{w}_k^J(\tilde{\mathbf{h}}(\bm{\theta}^*,\tilde{\mathbf{H}}), \bm{\lambda}^*)  ,\forall \tilde{\mathbf{H}} \} )$ obtained by Algorithm \ref{PDD_CSSCA_algorithm} is a stationary solution of the original problem \eqref{TTS_problem_detail_equi}, up to the error $O(e^J(\bm{\theta}^*, \bm{\lambda}^*))$ that diminishes to zero as $J$ goes to infinity.
\end{proof}

In Theorem \ref{theorem_3}, we have assumed that the initial step size $\gamma^0$ should be sufficiently small, this is to ensure that the whole step-size sequence $\{\gamma^t \}$ is sufficiently small since $\{\gamma^t\}$ is decreasing as specified in Step 3 of Section \ref{sec_long_term}. Note that this assumption is required to make it easier to handle the randomness caused by the random channels and for tractable and rigorous convergence analysis. In essence, it is a sufficient but not necessary condition to guarantee the convergence of Algorithm \ref{PDD_CSSCA_algorithm}, and in practice, we may prefer choosing a relatively large initial step size to achieve a faster initial convergence speed. Besides, in Theorem \ref{theorem_3}, we have assumed that a feasible initial point is available, which is also a sufficient (but not necessary) condition for convergence. Due to the feasible update in \eqref{long_term_problem_appro}, Algorithm \ref{PDD_CSSCA_algorithm} can still converge to a stationary point of problem \eqref{TTS_problem_detail_equi} with high probability, even when the initial point is infeasible \cite{Liu_CSSCA_2019}.

In the following, we analyze the computational complexity of the proposed PDD-TJAPB algorithm. First, it is observed that the complexity of the long-term IRS phase-shift optimization in Algorithm \ref{PDD_CSSCA_algorithm} is mainly due to computing $\{\mathbf{w}_k^J(\bm{\theta}^t, \bm{\lambda}^t,  \tilde{\mathbf{H}}_j^t)\}$ for the generated channel samples, calculating the gradients $\{\partial_{\bm{\theta}} \mathbf{w}^J(\tilde{\mathbf{h}}(\bm{\theta}^t,\tilde{\mathbf{H}}_j^t), \bm{\lambda}^t) \}$ and $\{ \partial_{\bm{\lambda}} \mathbf{w}^J(\tilde{\mathbf{h}}(\bm{\theta}^t,\tilde{\mathbf{H}}_j^t), \bm{\lambda}^t) \}$ through BP, and solving problem \eqref{long_term_problem_appro1} or \eqref{long_term_problem_appro}. For each $j \in \{1,\cdots,B \}$, Algorithm \ref{WMMSE_algorithm} is applied to obtain the corresponding active beamforming vectors, whose complexity is dominated by the matrix inversion operation required for updating $\{\mathbf{w}_k \}$ in \eqref{w_update}, which is $\mathcal{O}(JKM^3)$. The BP has in general the same computational complexity as the forward pass of the WMMSE network \cite{garcez2012neural}, therefore its complexity can also be expressed as $\mathcal{O}(JKM^3)$.
Besides, both problems \eqref{long_term_problem_appro1} and \eqref{long_term_problem_appro} can be expressed as a second-order cone program (SOCP) problem and the complexity of solving them using a standard interior-point method can be shown to be $\mathcal{O}(K^{0.5} (N+K)^3)$, where the basic elements of complexity analysis in \cite{Wang2014} are used. Accordingly, the complexity for updating the long-term IRS phase-shift matrix $\bm{\theta}$ is $\mathcal{O}(I(2BJKM^3 + K^{0.5} (N+K)^3))$, where $I$ denotes the iteration number required by Algorithm \ref{PDD_CSSCA_algorithm}.  Note that Algorithm \ref{PDD_CSSCA_algorithm} for long-term variable optimization only needs to be run once every $T_s \gg 1$ time slots. Second, for the short-term active beamforming optimization, the complexity per time slot can be expressed as $\mathcal{O}(JKM^3)$.

\begin{remark}
	\emph{The proposed PDD-TJAPB algorithm and the CSSCA framework in \cite{Liu_CSSCA_2019} can both handle stochastic optimization problems, where the objective and constraint functions are non-convex and involve expectations over random states. However, the difference is that the latter is designed to solve single-timescale stochastic optimization problems with only long-term optimization variables and cannot be directly applied to solve the formulated TTS beamforming optimization problem \eqref{TTS_problem_detail_equi} with both long-term and short-term variables. From this perspective, the proposed PDD-TJAPB algorithm is more powerful as a novel TTS primal-dual decomposition method is developed to decouple the long-term and short-term optimization variables and the deep unfolding technique is employed for gradient information extraction. To the best of our knowledge, the proposed PDD-TJAPB is the first efficient and provable convergent algorithm to solve a TTS stochastic optimization problem with both long-term and short-term variables coupled in non-convex stochastic constraints.
	}
\end{remark}

\begin{remark}
\emph{Note that for ease of implementation and cost reduction, the phase shift at each reflecting element of the IRS is usually restricted to a set of discrete values \cite{Wu2019Discrete, zhao2019intelligent}. In this case, we can first run the proposed PDD-TJAPB algorithm with continuous phase shifts; then, by quantizing the converged continuous phase shifts $\bm{\theta}$ to discrete values $\bm{\theta}^d$ (i.e., projecting $\bm{\theta}$ to $\mathcal{F}_d \triangleq  \{0,\frac{2\pi}{L},\cdots, \frac{2\pi(L-1)}{L}\}$, where $L=2^Q$  and $Q$ denotes the number of control bits for phase-shifting per IRS element), we rerun the PDD-TJAPB algorithm with fixed $\bm{\theta}^d$ to obtain the updated long-term Lagrange multipliers $\bm{\lambda}^d$; finally, we can follow the short-term part of the PDD-TJAPB algorithm with given $\bm{\theta}^d$ and $\bm{\lambda}^d$ to obtain the active beamforming vectors.}
\end{remark}

\section{Simulation Results} \label{section_simulation_results}
In this section, we provide numerical results to evaluate the performance of the proposed algorithm and draw useful insights. In our simulations, the distance-dependent path loss is modeled as $L = C_0\left(\frac{d_{\textrm{link}}}{D_0}\right)^{-\alpha}$, where $C_0$ is the path loss at the reference distance $D_0 = 1$ meter (m), $d_{\textrm{link}}$ represents the individual link distance and $\alpha$ denotes the path-loss exponent. Let $\alpha_{Au}$, $\alpha_{AI}$ and $\alpha_{Iu}$ denote the path-loss exponents of the AP-user, AP-IRS and IRS-user links, respectively, and let $L_{Au}$, $L_{AI}$ and $L_{Iu}$ denote the corresponding channel gains due to path loss. We set $\alpha_{Au} = 3.6$, $\alpha_{AI} = 2.2$ and $\alpha_{Iu} = 2.2$ \cite{Zhao2020intelligent}, i.e., the path-loss exponent of the AP-user link is larger than those of the AP-IRS and IRS-user links, to model the scenario that the users suffer from severe signal attenuation in the AP-user direct link.  
Besides, a three-dimensional coordinate system is considered, where the AP (equipped with a uniform linear array (ULA)) is located on the $x$-axis with antenna spacing $d_A = \lambda_c/2$ ($\lambda_c$ denotes the wavelength), while the IRS (equipped with a uniform planar array (UPA)) is located on the $y$-$z$ plane with reflecting element spacing $d_I = \lambda_c/8$. We set $N = N_yN_z$, where $N_y$ and $N_z$ denote the numbers of reflecting elements along the $y$-axis and $z$-axis, respectively. For the purpose of exposition, we fix $N_y = 4$. The reference antenna/element at the AP/IRS are located at $(2\;\textrm{m}, 0, 0)$ and $(0, 50\;\textrm{m}, 3\;\textrm{m})$, respectively. We assume that there are three users in the considered IRS-aided system, denoted by $U_k$'s, $k \in \{1,2,3\}$ and their locations are shown in Fig. \ref{fig:multiuser_setup}, i.e., the users are randomly located on a semicircle centered at $(2\;\textrm{m},50\;\textrm{m},0)$ with radius 3m. 

\begin{figure}[t] 
	\setlength{\abovecaptionskip}{-0.2cm} 
	\setlength{\belowcaptionskip}{-0.2cm} 
	\centering
	\scalebox{0.5}{\includegraphics{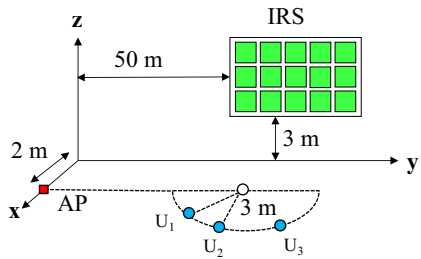}}
	\caption{Simulation setup of the multiuser system.}
	\label{fig:multiuser_setup} 
\end{figure}

We employ the general clustered delay line (CDL) models (in our simulations, the CDL-D model is used) for the AP-IRS, AP-user and IRS-uesr channels \cite[Table 7.7.1]{3GPPCDLD} and  flat fading is assumed. Specifically, the AP-IRS channel is modeled as 
\begin{equation}
\mathbf{G} = \sqrt{L_{AI}}\sum\limits_{q = 1}^Q {p(q) \mathbf{a}_T^q} {(\mathbf{a}_I^q)^H}, 
\end{equation}
where $Q$ denotes the total number of clusters; $p(q) \sim \mathcal{CN}(0,\sigma_q^2) $ denotes the complex gain of the $q$-th cluster\footnote{Note that ${p(1) \mathbf{a}_T^1} {(\mathbf{a}_I^1)^H}$ denotes the line-of-sight (LoS) component and $\sum\nolimits_{q = 2}^Q {p(q) \mathbf{a}_T^q} {(\mathbf{a}_I^q)^H}$ represents the non-LoS (NLoS)  component.};  $\mathbf{a}_T^q \in {\mathbb{C}^{M \times 1}}$ is the array response at the AP corresponds to the $q$-th cluster and ${[\mathbf{a}_T^q]_m} = {e^{ - \jmath 2 \pi (m - 1)d_A \sin {\theta_q^D} /\lambda_c }} ,\forall m$, where $\theta_q^D$ denotes the angle-of-departure (AoD); $\mathbf{a}_I^q \in {\mathbb{C}^{N \times 1}}$ is the array response of the $q$-th cluster at the IRS, with ${[\mathbf{a}_I^q]_n} = e^{ - \jmath 2\pi d_I(\lfloor n/{N_y} \rfloor \sin \varphi_q^A \sin {\theta_q^A} + (n - \lfloor n/{N_y} \rfloor N_y){\sin \varphi_q^A \cos {\theta_q^A})/\lambda_c}},\forall n$, where ${\theta_q^A}$ and $ \varphi_q^A$ denote the azimulth angle-of-arrival (AoA) and elevation AoA, respectively. In addition, we define the LoS power ratio $\beta_{AI}$ of the AP-IRS channel  as the power of the LoS component over the total channel power, which is given by $\beta_{AI} = \frac{|{p(1) \mathbf{a}_T^1} {(\mathbf{a}_I^1)^H}|^2}{\sum\nolimits_{q = 1}^Q |{p(q) \mathbf{a}_T^q} {(\mathbf{a}_I^q)^H}|^2}$. The AP-user and IRS-user channels are generated by following the similar procedure and the LoS power ratios of these two links are denoted by $\beta_{Au}$ and $\beta_{Iu}$, respectively. Other system parameters are set as follows unless otherwise specified: carrier frequency $f_c = 5$ GHz, $\sigma_k^2 = -80$ dBm,  $C_0 = -30$ dB, $N=40$, $M=6$, $K=3$, $B = 10$, $J=10$, $T_s=2000$, $R_k =R = 4$  bits/s/Hz, $\forall k$, $\tau_k =\tau = 0.01$, $\forall k$, $\rho^t = \frac{2}{{(2 + t)}^{0.9}}$, $\gamma^t = \frac{2}{2 + t}$ and the CDL-D channel parameters are set according to \cite[Table 7.7.1]{3GPPCDLD}. Note that the specific values for the coefficients in $\{\rho^t, \gamma^t\}$ such as $0.9$ and $2$ are tuned to achieve a good empirical convergence speed. For the initialization of the proposed PDD-TJAPB algorithm, the reflection phase shifts (i.e., $\bm{\theta}^0$) are independently and uniformly selected within $[0,2\pi)$, while the long-term Lagrange multipliers (i.e., $\bm{\lambda}^0$) are simply initialized as $\bm{\lambda}^0 = \mathbf{1}$.\footnote{We note that this simple initialization method cannot guarantee a feasible initial point in general; however, it works well in all of our simulations. Besides, theoretically, it is possible that the PDD-TJAPB algorithm may get stuck in an undesired point if the initial point is close to this undesired point. In this case, we can run the PDD-TJAPB algorithm with multiple random initial points, and it is very likely that the algorithm with one of the initial points will converge to a ``good'' stationary point of problem \eqref{TTS_problem_detail_equi}.} All the results are averaged over $2000$ independent channel realizations.

For comparison, we consider three benchmark schemes: 1) the fast-timescale scheme which employs the I-CSI and adopts the alternating optimization (AO) algorithm in \cite{Wu2018_journal} under short-term rate constraints of the users (with the same rate value as that for the case of long-term rate constraint for each user), 2) the TTS scheme in \cite{zhao2019intelligent}, where the weighted sum-rate maximization problem (the user weights are set to $1$) is solved for IRS phase-shift optimization and the active beamforming design problem in each time slot (with fixed IRS phase shifts) is solved by transforming it into an SOCP problem as in \cite{Bengtsson2001} that is then solved by CVX \cite{CVX} (this scheme is abbreviated as TTS-WSRMax in the following), and 3) the random phase-shift scheme where the phase shifts at the IRS are randomly generated at each time slot for introducing artificial channel fading to enhance the multiuser channel diversity, while the active beamforming is designed by adopting the same method as in the TTS-WSRMax scheme with fixed IRS phase shifts.

Prior to performance comparison, we first illustrate  the convergence behavior of the proposed PDD-TJAPB algorithm (i.e., Algorithm \ref{PDD_CSSCA_algorithm}). In Fig. \ref{fig:Convergence_behavior_compare1} (a) and (b), we plot the average transmit power at the AP and the average rate of the worst user (that has the lowest achievable average rate) versus the number of iterations, respectively. As can be seen, Algorithm \ref{PDD_CSSCA_algorithm} can converge within $200$ iterations. Due to the stochastic nature of the considered problem and the proposed algorithm, the power and rate curves in Fig. \ref{fig:Convergence_behavior_compare1} are not necessarily monotonic. However, from the overall trend, we observe that the average rate of the worst user gradually increases and finally approaches the  long-term average rate target $R=4$ bits/s/Hz. Besides, in Table \ref{tab:2}, we compare the channel estimation overhead and reflection coefficient signaling overhead (both measured by the number of coefficients) required by the TTS and conventional I-CSI-based schemes under a typical simulation setup as specified before (i.e., $M=6$, $K=3$, $N=40$ and $T_s=2000$). We assume that there is no prior knowledge of the channel and the popular least-square channel estimation method in \cite{zheng2019intelligent, you2019progressive, wang2019channel} is used, thus the channel estimation overhead is proportional to the number of channel coefficients.  It can be observed that by adopting the TTS scheme, both overheads can be significantly reduced as compared to the I-CSI-based scheme.

\begin{figure}[t]
	\setlength{\abovecaptionskip}{-0.2cm} 
	\setlength{\belowcaptionskip}{-0.2cm} 
	\centering
	\scalebox{0.38}{\includegraphics{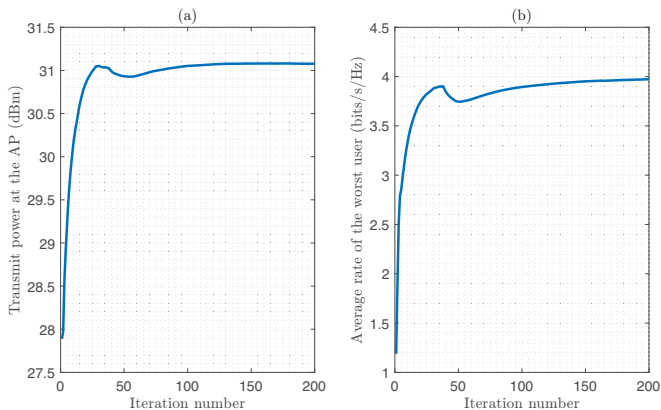}}
	\caption{Convergence behavior of Algorithm \ref{PDD_CSSCA_algorithm}.}
	\label{fig:Convergence_behavior_compare1}
\end{figure}

\begin{table*}[htbp]
	\centering
	\caption{Channel estimation overhead and reflection coefficient signaling overhead comparison}
	\label{tab:2}                        
	\begin{tabular}{ccccc} 
		\hline\noalign{\smallskip} 
		& TTS scheme & I-CSI-based scheme \\
		\noalign{\smallskip}\hline\noalign{\smallskip}
		Channel estimation overhead (per time slot) & 18 & 378 \\
	Reflection coefficient signaling overhead (per time interval) & 40 & 80000 \\
		\noalign{\smallskip}\hline
	\end{tabular}
\end{table*}

\begin{figure}[t] 
	\setlength{\abovecaptionskip}{-0.2cm} 
	\setlength{\belowcaptionskip}{-0.2cm} 
	\centering
	\scalebox{0.44}{\includegraphics{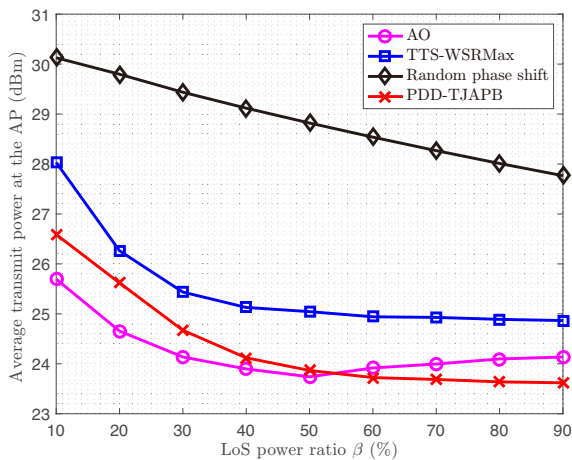}}
	\caption{Average transmit power versus LoS power ratio, $\beta$.}
	\label{fig:LoS_scale_factor}
\end{figure}

\subsection{Impact of LoS Power Ratio, $\beta$} In Fig. \ref{fig:LoS_scale_factor}, we investigate the average transmit power achieved by PDD-TJAPB versus the LoS power ratio of the AP-IRS and IRS-user links, where we set $\beta_{AI} = \beta_{Iu}   = \beta$ and $\beta_{Au} = 0$. First, we observe that the required average transmit power of the S-CSI-based schemes (i.e., TTS-WSRMax and PDD-TJAPB) and the random phase-shift scheme decreases with the increasing of $\beta$, and PDD-TJAPB achieves the best performance among them. Note that in these three schemes, the IRS phase shifts and the active beamforming vectors are not jointly designed in each time slot, therefore their ability to suppress the instantaneous multiuser interference is limited. When $\beta$ increases, their performance improves mainly because the multiuser interference becomes less random. For the S-CSI-based schemes, another important reason is that the AP-IRS-user link becomes more deterministic as $\beta$ increases, thus rendering the passive beamforming by TTS-WSRMax or PDD-TJAPB to be more effective. Second, it can be seen that the performance gain of PDD-TJAPB over TTS-WSRMax gradually enlarges as $\beta$ increases. This is reasonable since increasing $\beta$ will reduce the rank of $\mathbf{G}$, which further leads to decreased spatial multiplexing gain. Thus, more transmit power is required by the TTS-WSRMax scheme to guarantee the users' achievable rates in each time slot as only sum-rate maximization is considered while the user fairness is not taken into consideration. In contrast, by considering the achievable long-term  average rate constraint of each user, PDD-TJAPB is able to exploit the multiuser channel diversity more efficiently by performing ``time sharing'' based transmission scheduling for the users in different time slots. Therefore, better performance can be achieved, especially in the large $\beta$ regime. Finally, it is observed that the average transmit power of the AO scheme first decreases and then increases with the increasing of $\beta$. This is a joint effect of the better multiuser interference suppression capability of the AO scheme (as compared to the S-CSI-based schemes) and decreasing spatial multiplexing gain of the AP-IRS-user channel. Besides, due to the multiuser diversity gain, the performance of PDD-TJAPB is even better than that of AO when $\beta$ is larger than about $0.55$. This illustrates the advantage of employing long-term rate constraints over short-term rate constraints in exploiting the multiuser channel diversity.

\subsection{Impact of Number of Reflecting Elements, $N$}
Then, in Fig. \ref{fig:pow_IRS_number}, we investigate the average transmit power at the AP versus the number of reflecting elements at the IRS, $N$. We observe that the performance of all schemes improves with the increasing of $N$, which is reasonable since larger $N$ leads to higher aperture gain and reflect beamforming gain of the IRS in general. Besides, we can see that PDD-TJAPB achieves the lowest average transmit power and the performance gain over the other schemes increases with $N$. This is because when $N$ increases, the proposed PDD-TJAPB is able to more effectively reconfigure the S-CSI and hence better exploit the multiuser diversity gain.

\begin{figure}[t] 
	\setlength{\abovecaptionskip}{-0.2cm} 
	\setlength{\belowcaptionskip}{-0.2cm} 
	\centering
	\scalebox{0.44}{\includegraphics{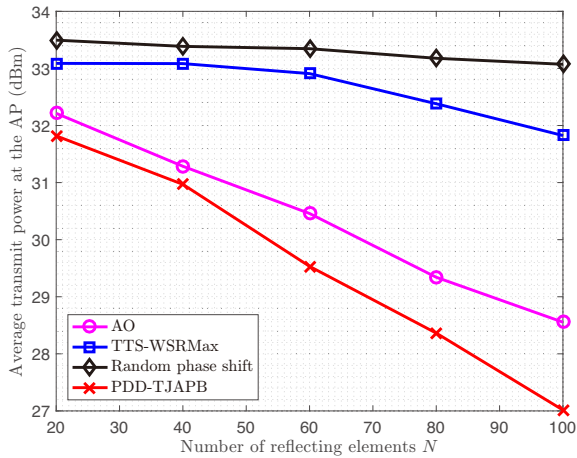}}
	\caption{Average transmit power versus number of reflecting elements, $N$.}
	\label{fig:pow_IRS_number} 
\end{figure}

\subsection{Impact of Discrete Phase Shifts}
In Fig. \ref{Fig_compareQ}, we show the impact of discrete phase shifts at the IRS on the average transmit power at the AP. In particular, we consider $Q = 1$, $Q = 2$ and $Q=3$ for discrete phase shifts at the IRS, and $Q=\infty$ denotes the (ideal) continuous phase-shift case. Other simulation parameters are the same as those in Fig. \ref{fig:pow_IRS_number}. From Fig. \ref{Fig_compareQ}, it is observed that as $Q$ increases, the performance gap between the discrete phase-shift case and the continuous phase-shift case ($Q =\infty $) gradually decreases and when $Q=3$, using IRS with discrete phase shifters incurs only negligible performance loss. Besides, it can be seen that this performance gap generally enlarges with the increasing of $N$ for given $Q$, which is reasonable since it is more difficult to obtain the performance of continuous phase shifts with quantized discrete phase shifts as $N$ becomes larger.

\begin{figure}[t] 
	\setlength{\abovecaptionskip}{-0.2cm} 
	\setlength{\belowcaptionskip}{-0.2cm} 
	\centering
	\scalebox{0.44}{\includegraphics{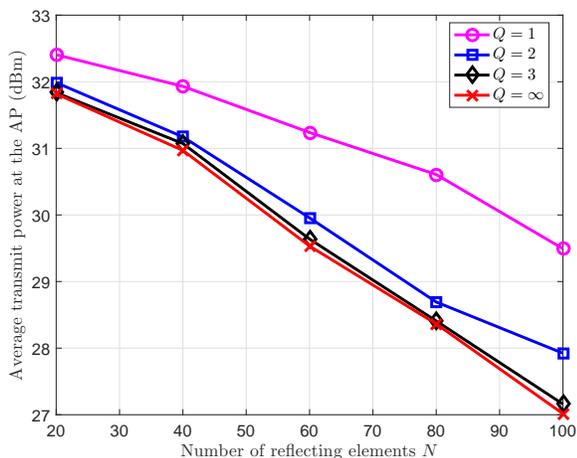}}
	\caption{Average transmit power versus  number of reflecting elements, $N$, with different values of $Q$.} 
	\label{Fig_compareQ}
\end{figure}

\subsection{Impact of CSI Delay} 
Next, in Fig. \ref{fig:CSI_delay}, we investigate the impact of CSI acquisition delay due to channel training/feedback on the average rate of the worst user and also compare the required average transmit power by different schemes, when $N=40$. Specifically, we consider time-varying channels over different time slots and thus the estimated CSI for each time slot can be outdated (although perfect channel estimation is assumed), i.e., there is generally CSI mismatch at the AP due to CSI delay caused by channel training/feedback. As such, there exists inevitable rate loss since the beamforming is designed based on the outdated CSI while the achievable rate is calculated based on the actual CSI. Similar to \cite{LiuMIMO2014}, we assume that the CSI delay is proportional to the dimension of the channel that is required at the AP, i.e., if the CSI delay for AP-user effective channel estimation (i.e., $\{ \tilde{\mathbf{h}}_k\}$) is $\omega$ millisecond (ms), then the CSI delay for estimating the full channel sample $\tilde{\mathbf{H}}$ is $\frac{{N K + M K + N M}}{{M K}} \omega$ ms. Besides, the outdated CSI is related to the actual CSI by the autoregressive model in \cite{Baddour2005} and the velocity of the users is set to $1$ kilometer per hour (km/h) as we mainly focus on the low-mobility scenario. It can be observed that the performance of all schemes degrades as the CSI delay increases. However, the performance degradation of the S-CSI-based schemes (TTS-WSRmax and PDD-TJAPB) is much smaller than the I-CSI-based scheme (AO) since the S-CSI-based schemes have smaller CSI delay. Besides, with similar average rate performance degradation, the required transmit power of PDD-TJAPB is smaller than that of TTS-WSRmax. On the other hand, though not shown in Fig. \ref{fig:CSI_delay}, we also note that the signaling overhead (for sending IRS reflection coefficients from the AP to IRS) of the S-CSI-based schemes is also much smaller than that of the I-CSI-based scheme. Therefore, the proposed scheme is more suitable for practical IRS-aided systems because it is more robust against CSI delays and has lower implementation cost.

\begin{figure}[t] 
	\setlength{\abovecaptionskip}{-0.2cm} 
	\setlength{\belowcaptionskip}{-0.2cm} 
	\centering
	\scalebox{0.44}{\includegraphics{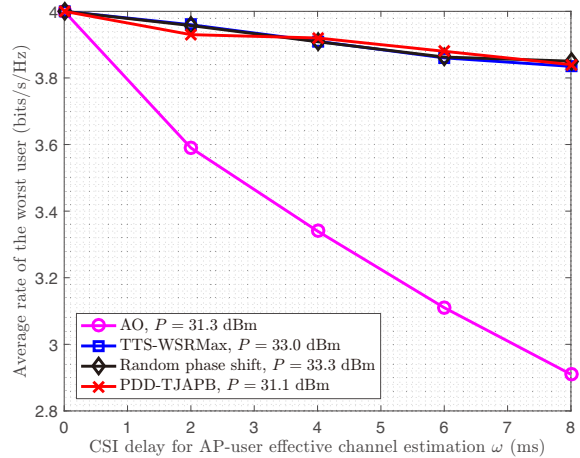}}
	\caption{Average rate of the worst user versus CSI delay for AP-user effective channel estimation, $\omega$.} 
	\label{fig:CSI_delay} 
\end{figure}

\subsection{Performance Comparison When Users Have Different Rate Targets} 
Finally, in Table \ref{tab:1}, we investigate the performance of the considered schemes when the users have different achievable average rate targets. For fair comparison, we fix the sum-rate of the users to $12$ bits/s/Hz, therefore when the users have the same rate target, we let $R_k=R=4$ bits/s/Hz; otherwise, we set $R_k=\{3,4,5\}$ bits/s/Hz. From Table \ref{tab:1}, we observe that for all the considered schemes, the required average transmit power increases when the users have different rate targets. Moreover, it is also observed that the average transmit power of PDD-TJAPB does not change much for these two cases, which shows that it is more suitable for the scenario when users have individual QoS requirements. 
 
\begin{table*}[htbp] 
\centering
\caption{Average transmit power comparison under different rate targets $\{R_k\}$} 
\label{tab:1}                        
	\begin{tabular}{ccccc} 
		\hline\noalign{\smallskip} 
		& AO & TTS-WSRMax & Random phase shift & PDD-TJAPB  \\
		\noalign{\smallskip}\hline\noalign{\smallskip}
		$R_k=R=4\;\textrm{bits/s/Hz},\forall k$ & 31.28 dBm & 33.00 dBm & 33.24 dBm & 31.16 dBm \\
		$R_k=\{3,4,5\}\;\textrm{bits/s/Hz}$ & 31.55 dBm & 33.67 dBm & 33.58 dBm & 31.33 dBm\\
		Average transmit power difference (dB) & 0.27 & 0.67 & 0.34 & 0.17 \\
		\noalign{\smallskip}\hline
	\end{tabular}
\end{table*}

\section{Conclusion}\label{section_conclusion}
In this paper, we studied the TTS beamforming optimization problem in an IRS-aided multiuser system under users' long-term average rate constraints.
The average transmit power at the AP was minimized by solving a TTS stochastic optimization problem, where the long-term IRS phase shifts are designed based on the S-CSI and the short-term active beamforming vectors are adaptively optimized according to the I-CSI of the users’ effective channels with given IRS phase shifts.
A new PDD-TJAPB algorithm was proposed by combining a novel TTS primal-dual decomposition method, the CSSCA framework and the deep unfolding technique. We proved that the proposed PDD-TJAPB algorithm is guaranteed to converge to a stationary solution of the considered problem almost surely. Compared with the existing I-CSI-based scheme, the proposed TTS and S-CSI based PDD-TJAPB not only leads to much lower signal processing complexity and channel training/signaling overhead, but also offers favorable performance under various practical setups. 

\begin{appendix}

	\subsection{Proof of Theorem \ref{theorem_2}} \label{appendix_b}
Let $\mathbf{w}^{J*} = \mathbf{w}^J(\tilde{\mathbf{h}}(\bm{\theta}^*,\tilde{\mathbf{H}}), \bm{\lambda}^*)$, we have the following equation according to the chain rule:
\begin{equation} \label{chain_rule}
\begin{aligned}
\nabla_{\bm{\lambda}}  (R_k & - r_k(\bm{\theta}^*, \mathbf{w}^{J*},\tilde{\mathbf{H}} )  ) \\
& =  \partial_{\bm{\lambda}} \mathbf{w}^{J*} \partial_{\mathbf{w}} (R_k \!-\!  r_k(\bm{\theta}, \mathbf{w}^{J*},\tilde{\mathbf{H}} ) ).
\end{aligned}
\end{equation}
Multiplying the vector inside the norm operator in \eqref{KKT_short_term} with $\partial_{\bm{\lambda}} \mathbf{w}^{J*}$, taking expectation w.r.t. $\tilde{\mathbf{H}}$, using \eqref{chain_rule} and letting $(\bm{\theta},\bm{\lambda}) = (\bm{\theta}^*,\bm{\lambda}^*)$, we have
\begin{equation} \label{KKT_short_term1}
\begin{aligned}
 \Bigg\|  \nabla_{\bm{\lambda}} & \mathbb{E}\Bigg\{ \sum\limits_{k\in \mathcal{K}}   \|\mathbf{w}^{J*} \|^2 \Bigg\}  + \sum\limits_{k \in \mathcal{K}} \lambda_k^* \nabla_{\bm{\lambda}} (R_k \\
& -\mathbb{E}\{ r_k(\bm{\theta},\mathbf{w}^{J*},\tilde{\mathbf{H}} )\} ) \Bigg\| = O(e^J(\bm{\theta}^*, \bm{\lambda}^*)).
\end{aligned}
\end{equation}
Combining \eqref{KKT_short_term1} with \eqref{KKT_2}, we obtain 
\begin{equation} \label{theorem_proof1}
\begin{aligned}
\Bigg\|  \sum\limits_{k \in \mathcal{K}} &
(\tilde{\lambda}_k - \lambda_k^*) \nabla_{\bm{\lambda}} (R_k \\
& -\mathbb{E}\{ r_k(\bm{\theta},\mathbf{w}^{J*},\tilde{\mathbf{H}} )\} ) \Bigg\| = O(e^J(\bm{\theta}^*, \bm{\lambda}^*)).
\end{aligned}
\end{equation}
Then, based on \eqref{theorem_proof1} and the LIRC, it follows that
\begin{equation} \label{theorem_proof2}
\tilde{\lambda}_k - \lambda_k^* = O(e^J(\bm{\theta}^*, \bm{\lambda}^*)), \;\forall k \in \mathcal{K}, 
\end{equation}
together with \eqref{KKT_3}, we can see that \eqref{stationary_3} is satisfied up to the error of $O(e^J(\bm{\theta}^*, \bm{\lambda}^*))$.

Similar to the derivation of \eqref{KKT_short_term1}, we have
\begin{equation} \label{theorem_proof3}
\begin{aligned}
  \Bigg\|\mathbb{E}& \left\{ \partial_{\bm{\theta}}  \mathbf{w}^{J*} \partial_{\mathbf{w}} \left(\sum\limits_{k\in \mathcal{K}}  \|\mathbf{w}^{J*}\|^2\right)\right\}  + \sum\limits_{k \in \mathcal{K}}
\lambda_k^* \mathbb{E} \{ \partial_{\bm{\theta}} \mathbf{w}^{J*} \\
& \times \partial_{\mathbf{w}}(R_k - r_k(\bm{\theta},\mathbf{w}^{J*},\tilde{\mathbf{H}} ) )\} \Bigg\| = O(e^J(\bm{\theta}^*, \bm{\lambda}^*)).
\end{aligned}
\end{equation}
Next, applying the chain rule to \eqref{KKT_1}, we obtain
\begin{equation} \label{theorem_proof4}
	\begin{aligned}
&  \mathbb{E} \left\{ \partial_{\bm{\theta}} \mathbf{w}^{J*} \partial_{\mathbf{w}} \sum\limits_{k\in \mathcal{K}} \|\mathbf{w}_k^{J^*}  \|^2 \right\}   \\
 & + \sum\limits_{k \in \mathcal{K}} \tilde{\lambda}_k \partial_{\bm{\theta}} (R_k - \mathbb{E} \{ r_k(\bm{\theta}^*, \mathbf{w}^{J*} ,\tilde{\mathbf{H}} ) \} ) \\
& + \sum\limits_{k \in \mathcal{K}} \tilde{\lambda}_k  \mathbb{E} \{ \partial_{\bm{\theta}} \mathbf{w}^{J*} \partial_{\mathbf{w}} (R_k - r_k(\bm{\theta}^*, \mathbf{w}^{J*} ,\tilde{\mathbf{H}} ) \} )
=\mathbf{0}.
\end{aligned}
\end{equation}
It follows from \eqref{theorem_proof2}, \eqref{theorem_proof3} and \eqref{theorem_proof4} that
\begin{equation}
	\begin{aligned}
\Bigg\|\sum\limits_{k \in \mathcal{K}} \lambda_k^* \partial_{\bm{\theta}} (R_k  - \mathbb{E} \{ r_k(\bm{\theta}^*, \mathbf{w}^{J*} ,& \tilde{\mathbf{H}} ) \} ) \Bigg\| \\
& = O(e^J(\bm{\theta}^*,\bm{\lambda}^*)),
\end{aligned}
\end{equation}
and together with \eqref{KKT_4}, we can see that \eqref{stationary_2} is satisfied up to the error of $O(e^J(\bm{\theta}^*,\bm{\lambda}^*))$.

Finally, it follows directly from \eqref{KKT_short_term} that \eqref{stationary_1} is satisfied up to the error of $O(e^J(\bm{\theta}^*, \bm{\lambda}^*))$. Therefore, the primal-dual pair $(\bm{\theta}^*, \{\mathbf{w}^{J*},\forall \tilde{\mathbf{H}} \}, \bm{\lambda}^*)$ satisfies the stationary conditions of problem \eqref{TTS_problem_detail_equi}. This thus completes the proof.

\subsection{Proof of Lemma \ref{lemma_1}} \label{appendix_c}
First, we can see that the following conditions hold: 1) the domains of $\bm{\theta}$ and $\mathbf{w}$ are compact and convex; 2) $  r_k(\bm{\theta},\{\mathbf{w}_k(\tilde{\mathbf{h}})\},\tilde{\mathbf{H}} ) $, $\forall k \in \mathcal{K}$ are real-valued and continuously differentiable functions of $\bm{\theta}$ and $\mathbf{w}$; 3) for any $k \in \mathcal{K}$, $\tilde{\mathbf{H}} \in \Omega$, the function
$r_k(\bm{\theta},\{\mathbf{w}_k(\tilde{\mathbf{h}})\},\tilde{\mathbf{H}} )$, its derivative w.r.t. $\bm{\theta}$ and $\mathbf{w}$, and its second-order
derivative w.r.t. $\bm{\theta}$ and $\mathbf{w}$, are all uniformly bounded. Then, together with the observations at the end of Section \ref{sec_short_term}, we can see that $\mathbf{f}_{\bm{\theta},k}^t$, $\mathbf{f}_{\mathbf{w},k}^t$ and $\mathbf{f}_{\bm{\lambda},k}^t$ are bounded. Therefore, properties 1-3 in Lemma \ref{lemma_1} follow directly from the expression of $\bar{f}_k^t(\bm{\theta}, \bm{\lambda})$ in \eqref{surrogate_function_update}. Property 4 is due to the chain rule and \cite[Lemma 1]{ruszczynski1980feasible}, for which the proof is similar to that of \cite[Lemma 1]{Yang2016TSP} and thus is omitted for conciseness. Finally, property 5 is a consequence of property 4. This thus completes the proof.
	\end{appendix}
\bibliographystyle{IEEETran}
\bibliography{references}

\begin{thebibliography}{10}
\providecommand{\url}[1]{#1}
\csname url@samestyle\endcsname
\providecommand{\newblock}{\relax}
\providecommand{\bibinfo}[2]{#2}
\providecommand{\BIBentrySTDinterwordspacing}{\spaceskip=0pt\relax}
\providecommand{\BIBentryALTinterwordstretchfactor}{4}
\providecommand{\BIBentryALTinterwordspacing}{\spaceskip=\fontdimen2\font plus
\BIBentryALTinterwordstretchfactor\fontdimen3\font minus
  \fontdimen4\font\relax}
\providecommand{\BIBforeignlanguage}[2]{{%
\expandafter\ifx\csname l@#1\endcsname\relax
\typeout{** WARNING: IEEEtran.bst: No hyphenation pattern has been}%
\typeout{** loaded for the language `#1'. Using the pattern for}%
\typeout{** the default language instead.}%
\else
\language=\csname l@#1\endcsname
\fi
#2}}
\providecommand{\BIBdecl}{\relax}
\BIBdecl

\bibitem{Wu2019Magazine}
Q.~{Wu} and R.~{Zhang}, ``Towards smart and reconfigurable environment:
  {I}ntelligent reflecting surface aided wireless network,'' \emph{IEEE Commun.
  Mag.}, vol.~58, no.~1, pp. 106--112, Jan. 2020.

\bibitem{Basar2019}
E.~{Basar}, M.~{Di Renzo}, J.~{De Rosny}, M.~{Debbah}, M.~{Alouini}, and
  R.~{Zhang}, ``Wireless communications through reconfigurable intelligent
  surfaces,'' \emph{IEEE Access}, vol.~7, pp. 116\,753--116\,773, Aug. 2019.

\bibitem{Huang2019}
C.~{Huang}, A.~{Zappone}, G.~C. {Alexandropoulos}, M.~{Debbah}, and C.~{Yuen},
  ``Reconfigurable intelligent surfaces for energy efficiency in wireless
  communication,'' \emph{IEEE Trans. Wireless Commun.}, vol.~18, no.~8, pp.
  4157--4170, Aug. 2019.

\bibitem{WuTutorial2020}
Q.~Wu, S.~Zhang, B.~Zheng, C.~You, and R.~Zhang, ``Intelligent reflecting
  surface aided wireless communications: {A} tutorial,'' \emph{IEEE Trans.
  Commun.}, DOI: 10.1109/TCOMM.2021.3051897, 2021.

\bibitem{Boccardi2014}
F.~{Boccardi}, R.~W. {Heath}, A.~{Lozano}, T.~L. {Marzetta}, and P.~{Popovski},
  ``Five disruptive technology directions for {5G},'' \emph{IEEE Commun. Mag.},
  vol.~52, no.~2, pp. 74--80, Feb. 2014.

\bibitem{Yang2019}
Y.~{Yang}, B.~{Zheng}, S.~{Zhang}, and R.~{Zhang}, ``Intelligent reflecting
  surface meets {OFDM:} {P}rotocol design and rate maximization,'' \emph{IEEE
  Trans. Commun.}, vol.~68, no.~7, pp. 4522--4535, Jul. 2020.

\bibitem{Cui2019}
M.~{Cui}, G.~{Zhang}, and R.~{Zhang}, ``Secure wireless communication via
  intelligent reflecting surface,'' \emph{IEEE Wireless Commun. Lett.}, vol.~8,
  no.~5, pp. 1410--1414, Oct. 2019.

\bibitem{Jiang2019}
T.~{Jiang} and Y.~{Shi}, ``Over-the-air computation via intelligent reflecting
  surfaces,'' in \emph{Proc. IEEE Global Communications Conference (GLOBECOM)},
  Dec. 2019, pp. 1--6.

\bibitem{wu2019joint}
Q.~Wu and R.~Zhang, ``Joint active and passive beamforming optimization for
  intelligent reflecting surface assisted {SWIPT} under {QoS} constraints,''
  \emph{IEEE J. Sel. Areas Commun.}, vol.~38, no.~8, pp. 1735--1748, Aug. 2020.

\bibitem{ZhangMIMO}
S.~Zhang and R.~Zhang, ``Capacity characterization for intelligent reflecting
  surface aided {MIMO} communication,'' \emph{IEEE J. Sel. Areas Commun.},
  vol.~38, no.~8, pp. 1823--1838, Aug. 2020.

\bibitem{lu2020enabling}
H.~{Lu}, Y.~{Zeng}, S.~{Jin}, and R.~{Zhang}, ``Enabling panoramic full-angle
  reflection via aerial intelligent reflecting surface,'' in \emph{IEEE
  International Conference on Communications Workshops (ICC Workshops)}, Jun.
  2020, pp. 1--6.

\bibitem{zuo2020resource}
J.~{Zuo}, Y.~{Liu}, Z.~{Qin}, and N.~{Al-Dhahir}, ``Resource allocation in
  intelligent reflecting surface assisted {NOMA} systems,'' \emph{IEEE Trans.
  Commun.}, vol.~68, no.~11, pp. 7170--7183, Nov. 2020.

\bibitem{Ding2020_IRSNOMA}
Z.~{Ding} and H.~{Vincent Poor}, ``A simple design of {IRS-NOMA}
  transmission,'' \emph{IEEE Commun. Lett.}, vol.~24, no.~5, pp. 1119--1123,
  May 2020.

\bibitem{PanJSAC2020}
C.~{Pan}, H.~{Ren}, K.~{Wang}, M.~{Elkashlan}, A.~{Nallanathan}, J.~{Wang}, and
  L.~{Hanzo}, ``Intelligent reflecting surface aided {MIMO} broadcasting for
  simultaneous wireless information and power transfer,'' \emph{IEEE J. Sel.
  Areas Commun.}, vol.~38, no.~8, pp. 1719--1734, Aug. 2020.

\bibitem{Pan2020TWC}
C.~{Pan}, H.~{Ren}, K.~{Wang}, W.~{Xu}, M.~{Elkashlan}, A.~{Nallanathan}, and
  L.~{Hanzo}, ``Multicell {MIMO} communications relying on intelligent
  reflecting surfaces,'' \emph{IEEE Trans. Wireless Commun.}, vol.~19, no.~8,
  pp. 5218--5233, Aug. 2020.

\bibitem{Atapattu2020Twoway}
S.~{Atapattu}, R.~{Fan}, P.~{Dharmawansa}, G.~{Wang}, J.~{Evans}, and T.~A.
  {Tsiftsis}, ``Reconfigurable intelligent surface assisted two-way
  communications: {Performance} analysis and optimization,'' \emph{IEEE Trans.
  Commun.}, vol.~68, no.~10, pp. 6552--6567, Oct. 2020.

\bibitem{PradhanWCL2020}
C.~{Pradhan}, A.~{Li}, L.~{Song}, B.~{Vucetic}, and Y.~{Li}, ``Hybrid precoding
  design for reconfigurable intelligent surface aided mmwave communication
  systems,'' \emph{IEEE Wireless Commun. Lett.}, vol.~9, no.~7, pp. 1041--1045,
  Jul. 2020.

\bibitem{Hu2018TSP}
S.~{Hu}, F.~{Rusek}, and O.~{Edfors}, ``Beyond massive {MIMO}: {The} potential
  of data transmission with large intelligent surfaces,'' \emph{IEEE Trans.
  Signal Process.}, vol.~66, no.~10, pp. 2746--2758, May 2018.

\bibitem{Mishra2019ICASSP}
D.~{Mishra} and H.~{Johansson}, ``Channel estimation and low-complexity
  beamforming design for passive intelligent surface assisted {MISO} wireless
  energy transfer,'' in \emph{Proc. IEEE International Conference on Acoustics,
  Speech and Signal Processing (ICASSP)}, May 2019, pp. 4659--4663.

\bibitem{zheng2019intelligent}
B.~{Zheng} and R.~{Zhang}, ``Intelligent reflecting surface-enhanced {OFDM}:
  {C}hannel estimation and reflection optimization,'' \emph{IEEE Wireless
  Commun. Lett.}, vol.~9, no.~4, pp. 518--522, Apr. 2020.

\bibitem{you2019progressive}
C.~{You}, B.~{Zheng}, and R.~{Zhang}, ``Channel estimation and passive
  beamforming for intelligent reflecting surface: {D}iscrete phase shift and
  progressive refinement,'' \emph{IEEE J. Sel. Areas Commun.}, vol.~38, no.~11,
  pp. 2604--2620, Nov. 2020.

\bibitem{wang2019channel}
Z.~{Wang}, L.~{Liu}, and S.~{Cui}, ``Channel estimation for intelligent
  reflecting surface assisted multiuser communications: {F}ramework,
  algorithms, and analysis,'' \emph{IEEE Trans. Wireless Commun.}, vol.~19,
  no.~10, pp. 6607--6620, Oct. 2020.

\bibitem{zheng2020intelligent}
B.~{Zheng}, C.~{You}, and R.~{Zhang}, ``Intelligent reflecting surface assisted
  multi-user {OFDMA}: {C}hannel estimation and training design,'' \emph{IEEE
  Trans. Wireless Commun.}, vol.~19, no.~12, pp. 8315--8329, Dec. 2020.

\bibitem{He2019_CE}
Z.~{He} and X.~{Yuan}, ``Cascaded channel estimation for large intelligent
  metasurface assisted massive {MIMO},'' \emph{IEEE Wireless Commun. Lett.},
  vol.~9, no.~2, pp. 210--214, Feb. 2020.

\bibitem{chen2019channel}
J.~Chen, Y.-C. Liang, H.~V. Cheng, and W.~Yu, ``Channel estimation for
  reconfigurable intelligent surface aided multi-user {MIMO} systems,''
  \emph{arXiv preprint arXiv:1912.03619}, 2019.

\bibitem{Wu2018_journal}
{Q. {Wu} and R. {Zhang}}, ``Intelligent reflecting surface enhanced wireless
  network via joint active and passive beamforming,'' \emph{IEEE Trans.
  Wireless Commun.}, vol.~18, no.~11, pp. 5394--5409, Nov. 2019.

\bibitem{Wu2019Discrete}
Q.~{Wu} and R.~{Zhang}, ``Beamforming optimization for wireless network aided
  by intelligent reflecting surface with discrete phase shifts,'' \emph{IEEE
  Trans. Commun.}, vol.~68, no.~3, pp. 1838--1851, Mar. 2020.

\bibitem{Guo2020TWC}
H.~{Guo}, Y.~{Liang}, J.~{Chen}, and E.~G. {Larsson}, ``Weighted sum-rate
  maximization for reconfigurable intelligent surface aided wireless
  networks,'' \emph{IEEE Trans. Wireless Commun.}, vol.~19, no.~5, pp.
  3064--3076, May 2020.

\bibitem{Di2020JSAC}
B.~{Di}, H.~{Zhang}, L.~{Song}, Y.~{Li}, Z.~{Han}, and H.~V. {Poor}, ``Hybrid
  beamforming for reconfigurable intelligent surface based multi-user
  communications: {Achievable} rates with limited discrete phase shifts,''
  \emph{IEEE J. Sel. Areas Commun.}, vol.~38, no.~8, pp. 1809--1822, Aug. 2020.

\bibitem{Zhao2020intelligent}
M.~M. Zhao, Q.~Wu, M.~J. Zhao, and R.~Zhang, ``Exploiting amplitude control in
  intelligent reflecting surface aided wireless communication with imperfect
  {CSI},'' \emph{IEEE Trans. Commun.}, DOI: 10.1109/TCOMM.2021.3064959, 2021.

\bibitem{zhao2019intelligent}
{M. M. Zhao}, Q.~Wu, {M. J. Zhao}, and R.~Zhang, ``Intelligent reflecting
  surface enhanced wireless network: {T}wo-timescale beamforming
  optimization,'' \emph{IEEE Trans. Wireless Commun.}, vol.~20, no.~1, pp.
  2--17, Jan. 2021.

\bibitem{Liu_CSSCA_2019}
A.~{Liu}, V.~K.~N. {Lau}, and B.~{Kananian}, ``Stochastic successive convex
  approximation for non-convex constrained stochastic optimization,''
  \emph{IEEE Trans. Signal Process.}, vol.~67, no.~16, pp. 4189--4203, Aug.
  2019.

\bibitem{CVX}
M.~Grant and S.~Boyd, ``{CVX}: Matlab software for disciplined convex
  programming, version 2.1,'' \url{http://cvxr.com/cvx}, Mar. 2014.

\bibitem{Zhang2018Metasurface}
{L. {Zhang} \emph{et al.}}, ``Space-time-coding digital metasurfaces,''
  \emph{Nat. Commun.}, vol.~9, no.~1, p. 4334, Oct. 2018.

\bibitem{Mestre2008}
X.~{Mestre}, ``Improved estimation of eigenvalues and eigenvectors of
  covariance matrices using their sample estimates,'' \emph{IEEE Trans. Inf.
  Theory}, vol.~54, no.~11, pp. 5113--5129, Nov. 2008.

\bibitem{Werner2008}
K.~{Werner}, M.~{Jansson}, and P.~{Stoica}, ``On estimation of covariance
  matrices with kronecker product structure,'' \emph{IEEE Trans. Signal
  Process.}, vol.~56, no.~2, pp. 478--491, Feb. 2008.

\bibitem{Liu2015TWC}
A.~{Liu} and V.~K.~N. {Lau}, ``Two-stage subspace constrained precoding in
  massive {MIMO} cellular systems,'' \emph{IEEE Trans. Wireless Commun.},
  vol.~14, no.~6, pp. 3271--3279, Jun. 2015.

\bibitem{Khojastepour2004}
M.~A. {Khojastepour} and B.~{Aazhang}, ``The capacity of average and peak power
  constrained fading channels with channel side information,'' in \emph{IEEE
  Wireless Communications and Networking Conference (WCNC)}, vol.~1, Mar. 2004,
  pp. 77--82.

\bibitem{Shi2011WMMSE}
Q.~{Shi}, M.~{Razaviyayn}, Z.~{Luo}, and C.~{He}, ``An iteratively weighted
  {MMSE} approach to distributed sum-utility maximization for a {MIMO}
  interfering broadcast channel,'' \emph{IEEE Trans. Signal Process.}, vol.~59,
  no.~9, pp. 4331--4340, Sep. 2011.

\bibitem{zhou2020framework}
G.~{Zhou}, C.~{Pan}, H.~{Ren}, K.~{Wang}, and A.~{Nallanathan}, ``A framework
  of robust transmission design for {IRS}-aided {MISO} communications with
  imperfect cascaded channels,'' \emph{IEEE Trans. Signal Process.}, vol.~68,
  pp. 5092--5106, 2020.

\bibitem{hershey2014deep}
J.~R. Hershey, J.~L. Roux, and F.~Weninger, ``Deep unfolding: {M}odel-based
  inspiration of novel deep architectures,'' \emph{arXiv preprint
  arXiv:1409.2574}, 2014.

\bibitem{rumelhart1986learning}
D.~E. Rumelhart, G.~E. Hinton, and R.~J. Williams, ``Learning representations
  by back-propagating errors,'' \emph{Nature}, vol. 323, no. 6088, pp.
  533--536, 1986.

\bibitem{bishop2006pattern}
C.~M. Bishop, \emph{Pattern recognition and machine learning}.\hskip 1em plus
  0.5em minus 0.4em\relax springer, 2006.

\bibitem{garcez2012neural}
A.~S.~d. Garcez, K.~B. Broda, and D.~M. Gabbay, \emph{Neural-symbolic learning
  systems: {F}oundations and applications}.\hskip 1em plus 0.5em minus
  0.4em\relax Springer Science \& Business Media, 2012.

\bibitem{Wang2014}
K.~{Wang}, A.~M. {So}, T.~{Chang}, W.~{Ma}, and C.~{Chi}, ``Outage constrained
  robust transmit optimization for multiuser {MISO} downlinks: {T}ractable
  approximations by conic optimization,'' \emph{IEEE Trans. Signal Process.},
  vol.~62, no.~21, pp. 5690--5705, Nov. 2014.

\bibitem{3GPPCDLD}
3GPP, ``Study on channel model for frequencies from 0.5 to 100 {GHz} ({3GPP TR}
  38.901 version 16.1.0 release 16),'' Dec. 2019. [Online]. Available:
  https://www.3gpp.org/ftp/Specs/archive/38\_series/38.901/38901-g10.zip.
  [Accessed 7 Jun. 2020].

\bibitem{Bengtsson2001}
M.~Bengtsson and B.~Ottersten, ``Optimal and suboptimal transmit beamforming,''
  \emph{Handbook of Antennas in Wireless Communications}, 2001.

\bibitem{LiuMIMO2014}
A.~{Liu} and V.~{Lau}, ``Phase only {RF} precoding for massive {MIMO} systems
  with limited {RF} chains,'' \emph{IEEE Trans. Signal Process.}, vol.~62,
  no.~17, pp. 4505--4515, Sep. 2014.

\bibitem{Baddour2005}
K.~E. {Baddour} and N.~C. {Beaulieu}, ``Autoregressive modeling for fading
  channel simulation,'' \emph{IEEE Trans. Wireless Commun.}, vol.~4, no.~4, pp.
  1650--1662, Jul. 2005.

\bibitem{ruszczynski1980feasible}
A.~Ruszczy{\'n}ski, ``Feasible direction methods for stochastic programming
  problems,'' \emph{Math. Programm.}, vol.~19, no.~1, pp. 220--229, Dec. 1980.

\bibitem{Yang2016TSP}
Y.~{Yang}, G.~{Scutari}, D.~P. {Palomar}, and M.~{Pesavento}, ``A parallel
  decomposition method for nonconvex stochastic multi-agent optimization
  problems,'' \emph{IEEE Trans. Signal Process.}, vol.~64, no.~11, pp.
  2949--2964, Jun. 2016.

\end{thebibliography}

\end{document}